\documentclass[11pt]{article}
\usepackage{booktabs}
\usepackage[ruled]{algorithm2e}

\SetAlFnt{\small}
\SetAlCapFnt{\small}
\SetAlCapNameFnt{\small}
\SetAlCapHSkip{0pt}
\IncMargin{-\parindent}

\usepackage{fullpage}
\usepackage{amsmath, amsthm, amssymb}
\usepackage{verbatim}
\usepackage[sort]{natbib}
\usepackage{xcolor}
\usepackage{hyperref}

\usepackage{authblk}

\newtheorem{proposition}{Proposition}
\newtheorem{lemma}{Lemma}
\newtheorem{theorem}{Theorem}
\newtheorem{corollary}{Corollary}
\newtheorem{assumption}{Assumption}
\theoremstyle{definition}
\newtheorem{definition}{Definition}
\newtheorem{example}{Example}

\newcommand{\cO}{\mathcal{O}}
\newcommand{\cV}{\mathcal{V}}
\newcommand{\cP}{\mathcal{P}}
\newcommand{\bR}{\mathbb{R}}
\newcommand{\eR}{\overline{\mathbb{R}}_+}

\newcommand{\eps}{\varepsilon}

\newcommand{\pri}{\mathrm{int}}
\newcommand{\pub}{\mathrm{ext}}

\newcommand{\altIC}{AltPIC}
\newcommand{\altNB}{AltNB}

\title{Nonbossy Mechanisms: Mechanism Design Robust to Secondary Goals}

\author[a]{Renato Paes Leme}
\author[a]{Jon Schneider}
\author[b]{Hanrui Zhang}

\affil[a]{Google Research, 111 8th Avenue, New York, NY 10011, USA}
\affil[b]{Chinese University of Hong Kong, Shatin, NT, Hong Kong SAR, China}
\affil[ ]{\texttt{\{renatoppl,jschnei\}@google.com}, \texttt{hanrui@cse.cuhk.edu.hk}}

\date{}

\begin{document}

\maketitle

\begin{abstract}
    We study mechanism design when agents may have hidden secondary goals which will play a role when the primary utility of the outcomes is the same.
    We show that in such cases, a mechanism is immune to strategic manipulation if and only if it is incentive compatible with regard to primary utility --- a property we term ``primary incentive compatibility'' --- {\em and} nonbossy --- a well-studied property in the context of matching and allocation mechanisms.
    We give complete characterizations of primarily incentive-compatible and nonbossy mechanisms in various settings, including auctions with single-parameter agents and public decision settings where all agents share a common outcome.
    In particular, we show that in the single-item setting, a mechanism is primarily incentive compatible, individually rational, and nonbossy if and only if it is a sequential posted-price mechanism.
\end{abstract}



\noindent {\bf Declaration of interest:} none

\section{Introduction}

Incentive-compatible direct mechanisms guarantee that it is in the best interest of each {\em self-interested} agent to report their true type to the mechanism.
However, such mechanisms may still be prone to manipulation by agents who have secondary goals not expressed in their utility function. Consider for example a broker that represents two different buyers $A$ and $B$ in an auction. Given their fiduciary duties, the broker is obliged to get the best possible outcome for each client. When choosing between two bids on behalf of client $A$ that lead to the same outcome for client $A$, the broker may break ties in favor of the bid that benefits client $B$. Such ``hidden secondary goals'' complicate the standard picture of mechanism design.

Interestingly, many common mechanisms that are incentive compatible in the classical sense are no longer immune to strategic behavior in the presence of these secondary goals. For example, consider two ``weakly cooperating'' (in the sense explained below) agents participating in a second price auction, where a single item is being sold.
These agents behave in the following way: their primary goal is to maximize their own utility, which is their private value minus the payment if they win the item, and $0$ otherwise; in the case where they cannot (rationally) win the item, they instead try to achieve their secondary goal, which is to maximize the other agent's utility.
Then, assuming both agents know each other's value, one natural bidding strategy for both agents is the following: the agent with the higher value should bid their true value, and the agent with the lower value should bid $0$.
As a result, the agent with the higher value still wins the item, but their payment is always $0$.
In fact, the ``intended behavior'', i.e., both agents bidding their true values, is not even an equilibrium in this case, because the agent with the lower value would have strict incentive to deviate and bid $0$.

On the other hand, there do exist some incentive-compatible mechanisms (again, in the classical sense) which remain robust\footnote{
    Throughout the paper, we use ``robustness'' to refer to properties that in essence capture some kind of incentive compatibility.  We leave the term ``(primary) incentive compatibility'' for more specific notions of robustness to avoid ambiguity.
}
in the presence of secondary goals. One reason why sequential posted-price mechanisms --- where an auctioneer iterates through the agents in order, offering each agent a take-it-or-leave-it price --- are appealing in practice is this type of robustness.\footnote{
    Although we describe sequential posted-price mechanisms in the common, indirect form, the actual mechanisms that we refer to throughout the paper are the {\em direct implementation} of sequential posted-price mechanisms.
    Such a direct implementation would solicit valuations from all agents, compare these reported valuations with the prices in the same order that the agents are visited, and serve the first agent whose reported valuation is higher than the price.
}
In particular, in a posted-price mechanism, an agent cannot change the utility of other agents without changing their own outcome (winning or losing the item), and thus cannot optimize their secondary objective without losing primary utility.\footnote{
    Here we assume agents always have strict preferences between different alternatives even if they happen to induce the same quasi-linear utility.
}
Note that posted-price mechanisms are further agnostic to the specific types of secondary goals of the agents: they are robust regardless of whether the agents' secondary goals are altruistic or malicious (or completely arbitrary), and do not require the agents to report their secondary goals upfront.

The above examples illustrate the setting that we call ``hidden secondary goals'' where secondary goals are present and unreported to the mechanism designer, but are less important than an agent's primary goal of maximizing their own utility. It is conceivable that ``rational agents'' (in a loose sense) sometimes do behave this way: such agents act carefully to optimize for their main interest, {\em except when it does not really matter} --- if multiple actions are equally good as far as their main interest is concerned, then they turn to further optimizing for their interest in other agents, be it sympathy, envy, or any other consideration deemed rationally irrelevant. This is especially common in cases where agents participate in the auction via brokers or bidding agents, such as real-estate brokers, literary agents, art dealers, ... Brokers are contractually or legally obligated to act in their clients' best interest, but may follow hidden secondary goals when this does not affect those interests.
The secondary goals are {\em hidden}, in the sense that the mechanism designer, and perhaps even the agents, are unsure about their existence beforehand --- in fact, it may very well be the case that some agents have secondary goals while others do not.
As a result, it would make little sense to elicit such secondary goals explicitly in the mechanism: on one hand, agents who happen to be free of secondary goals would simply be confused; on the other, the potential benefits of doing so would hardly justify the effort, given how rich these secondary goals generally can be.
In fact, in domains such as matching, it is a common practice not to elicit full preferences and rely solely on, for example, ordinal information.

In this paper, we study the problem of mechanism design in the setting of hidden secondary goals --- such preference models have already been considered in other contexts including decision theory \citep{fishburn1975axioms,blume1991lexicographic}, resource allocation \citep{luss1999equitable,saban2014note}, learning \citep{booth2010learning,diana2021lexicographically}, and inference \citep{kohli2007representation}. Our main goal is to answer the following questions:
\begin{quote}
    \em What mechanisms are incentive compatible when agents exhibit hidden secondary goals?
\end{quote}

\subsection{Our Results}

\paragraph{Alternate interpretation: nonbossy mechanisms.}
We begin by giving another interpretation of incentive compatibility\footnote{We focus on {\em dominant-strategy} incentive compatibility throughout the paper, both in the classical sense and in the presence of secondary goals.} in the presence of secondary goals: incentive compatibility can be decomposed into the following two properties:
\begin{itemize}
    \item {\em ``Primary'' incentive compatibility} (henceforth PIC for short): that no agent can benefit in terms of their {\em primary} goal by misreporting their utility function.
    This aligns with the classical notion of incentive compatibility in the absence of secondary goals, with regard to primary utility only.
    \item {\em Nonbossiness}: that no agent can change the outcomes that other agents receive, or the amounts that other agents pay, without changing their own outcome or payment (in accordance with the existing notion of nonbossiness in the economics literature; see Section \ref{sec:related} for an overview).
\end{itemize}
Moreover, PIC and nonbossiness together are also sufficient for incentive compatibility in the presence of secondary goals.
This interpretation of incentive-compatible mechanisms with secondary goals as nonbossy mechanisms is more amenable to analysis and enables all subsequent results of the paper.

\paragraph{Characterizing nonbossy mechanisms.}
We then turn to the structure of nonbossy mechanisms in various concrete classes of environments that are commonly considered in the context of mechanism design.
We first prove a general characterization of the payment rules of nonbossy mechanisms. Specifically, we show that as long as the environment is ``expressive'' enough (containing a wide enough range of valuation functions), the payment rule of any PIC and nonbossy mechanism is a function of the outcomes {\em only} --- if two valuation profiles lead to the same outcomes for all agents, then the amounts that all agents pay must also be the same. Almost all commonly studied environments are expressive enough for this characterization to hold\footnote{We also show this property is necessary for this characterization, in the sense that there exist environments without this property in which the payment characterization fails.}. 

Based on the payment characterization, we further investigate two prominent settings in mechanism design, and provide complete characterizations of nonbossy mechanisms.
First, we consider the abstract setting where the mechanism selects a single outcome (corresponding to a public decision, such as the location of a hospital) that is shared by all agents, and each agent may value each potential outcome arbitrarily.
For this setting, we show that quite surprisingly, any mechanism that is PIC and nonbossy must either be dictatorial or always choose between only $2$ outcomes.
In other words, when choosing a common outcome, only trivial mechanisms can be immune to strategic manipulation if agents care (even infinitesimally) about each other.
This is reminiscent of the celebrated Gibbard-Satterthwaite theorem~\citep{gibbard1973manipulation,satterthwaite1975strategy}, and indeed the proof works by reducing to that theorem.

We then investigate the single-parameter\footnote{
    Throughout the paper, the term ``single-parameter'' refers to settings where agents each have a single-dimensional type, possibly complicated by an additional tiebreaking ordering over possible outcomes.
} setting where identical items are allocated, and each agent is interested in at most one item.
For this setting, we show that nonbossy (and PIC and individually rational) mechanisms generalize sequential posted-price mechanisms, which approach all agents one by one and make take-it-or-leave-it offers.
In particular, if there is a single item being sold, then the two classes of mechanisms are exactly the same.
This provides an interesting semantic interpretation of sequential posted-price mechanisms in the single-item setting: the only mechanisms in the single-item setting that are IC in the presence of secondary goals are sequential posted-price mechanisms.
To the best of our knowledge, this is the first characterization of this kind for sequential posted-price mechanisms.

When there are multiple items (or even richer feasibility constraints), nonbossy, PIC and individually rational mechanisms are strictly richer than sequential posted-price mechanisms: in particular, they allow for posting a {\em vector} of prices to all agents simultaneously, which is accepted if and only if each agent accepts their personal component of the price vector.
More generally, nonbossy, PIC and individually rational mechanisms correspond to {\em decision lists with exceptions}: a mechanism can be implemented using a sequence of price vectors, each associated with a distinct vector of outcomes and an exception list, consisting of a subset of other price vectors.
The mechanism checks these price vectors one by one, and chooses the outcomes associated with the first vector satisfying: (1) all agents accept the price vector, and (2) there does not exist another price vector that is also accepted in the exception list of this price vector.
In fact, we show that if the exception lists are chosen appropriately, then the order in which these price vectors are considered does not matter.
The proof explicitly constructs a decision list given a mechanism, by first showing that any nonbossy and PIC mechanism must ``break ties consistently'', i.e., if the mechanism chooses one outcome vector when another outcome vector is also satisfied, then it can never choose the latter vector whenever the former is satisfied.
Given this, one can construct a decision tree from a mechanism, which can then be turned into a decision list by merging nodes corresponding to the same outcome vector.

\subsection{Related Work}\label{sec:related}

\paragraph{Nonbossy allocation rules and social choice functions.}

Our results can be viewed as a monetary version of previous results on nonbossy allocation rules and social choice functions.
\citet{satterthwaite1981strategy} propose the notion of nonbossiness for allocation mechanisms without money, and give several characterizations thereof.
Since then, nonbossy allocation rules have been studied extensively: see the survey by \citet{thomson2016non} for a comprehensive exposition.
The vast majority of these results focus on environments without money (e.g., allocation of indivisible items under unit demand~\citep{pycia2017incentive,bade2020random} or multi-unit demand~\citep{svensson1999strategy,papai2001strategyproof,hatfield2009strategy}, stable two-sided matching~\citep{kojima2010impossibility},\footnote{\citet{bade2020random} provides an alternative proof for the characterization of individually strategyproof and nonbossy mechanisms by \citet{pycia2017incentive}.} etc.), which make them incomparable with our results.
Below we discuss several results that are particularly related to ours.

\citet{svensson2002strategy} study nonbossy allocation mechanisms where, in addition to indivisible items, there is a fixed amount of money being allocated.
They show that the space of strategyproof and nonbossy mechanisns is finite, and that this space can be further restricted with additional properties enforced.
This setting is different from ours, because the total ``payment'' (i.e., money being allocated) has to sum up to the fixed amount --- indeed, the finiteness of the space of mechanisms is too strong to be true with general payments that we consider.

\citet{mishra2014non} study single-item auctions with a property which \citet{thomson2016non} terms ``non-monetary nonbossiness'': the property states that no agent should be able to change other agents' outcomes without changing their own outcome, but puts no restrictions on payments.
\citet{mishra2014non} show an allocation rule is implementable in dominant strategies and non-monetary nonbossy if and only if it is strongly rationalizable.
\citet{nath2015affine} study ``allocation nonbossy'' social choice functions, which is the same as non-monetary nonbossiness.
They show that given some richness assumptions, any social choice function that is onto, strategyproof and non-monetary nonbossy must be an affine maximizer.
\citet{mukherjee2015axioms} studies ``nonbossiness in decision'', which is again equivalent to non-monetary nonbossiness.
The author shows that in the single-item setting, any anonymous, strategyproof, and non-monetary nonbossy allocation rule must be that of a Vickrey auction with a reserve price.
As noted by \citet{thomson2016non}, non-monetary nonbossiness is incomparable to the notion of nobossiness that we study.
Moreover, it is not clear what semantic interpretations (e.g., robustness against a certain form of strategic behavior) non-monetary nonbossiness has beyond its very definition.

In the context of the Shapley-Scarf house allocation model~\citep{shapley1974cores}, \citet{miyagawa2001house} shows that a mechanism is strategyproof, individually rational, and nonbossy iff it is a fixed-price mechanism, and \citet{svensson2002fixed} gives an alternative, shorter proof.  Established in a different setting, this result is not directly comparable to our characterizations.  Nonetheless, it is connected to our characterization in the single-item setting in the following sense: in the single-item setting, a fixed-price mechanism sets upfront a fixed price for each agent, elicits agents' valuations, and then allocates the item through the core assignment, in which the item goes to the agent with the highest price among those who are willing to pay their respective prices.  These mechanisms are subsumed by sequential posted-price mechanisms --- in fact, they are a restricted class of sequential posted-price mechanisms, where the principal must visit agents with higher prices first.  Compared to fixed-price mechanisms, sequential posted-price mechanisms exhibit one additional degree of freedom, where the mechanism specifies not only the prices offered to the agents, but also the orders in which these offers are made.

Also worth noting is the recent result by \citet{pycia2025ordinal}.  In particular, their Theorem~5 states that, in the unit-demand house allocation setting, a mechanism is strategyproof and Arrovian efficient with respect to a complete social choice function, iff it is an ``almost sequential dictatorship'', which roughly works by iteratively picking one unmatched agent in each round, and let that agent pick their most preferred house among the unmatched ones, with the possible exception that, when only two houses are left unmatched and the two remaining agents happen to prefer the same house, the mechanism gets to break the tie depending on which house is commonly preferred, instead of choosing one of the two agents to be the next dictator.  This is in the same vein as our characterization for the common outcome setting, which involves two similar conditions: dictatorship and choosing between two outcomes.  In fact, nonbossiness is implied in this particular result by \citet{pycia2025ordinal}.

Finally, perhaps most closely related to our results is the recent work by \citet{pycia2021non}, who study the same notion of nonbossiness as ours in single-item auctions with general payments.
They show, in a Bayesian setting, that (1) the first-price auction with no reserve is the essentially unique mechanism that is nonbossy, individually rational, and efficient in equilibrium, and (2) the first-price auction with the optimal reserve price is the essentially unique mechanism that is nonbossy, individually rational, and revenue maximizing.
These results are not directly comparable to ours: aiming to capture dominant-strategy robustness against strategic behavior with secondary goals, we focus on nonbossy mechanisms that are also (dominant-strategy) PIC, which first-price auctions are not.\footnote{
    Note that with multiple agents, the revelation principle generally only gives {\em Bayesian} incentive compatibility, which is generally weaker than dominant-strategy (primary) incentive compatibility that we consider.
}
In fact, both the results by \citet{pycia2021non} and our results imply that no PIC and nonbossy mechanism can be efficient (i.e., welfare-maximizing) or revenue-maximizing in the single-item setting.

\paragraph{Stronger notions of robustness.} 
Conceptually, our results are along the line of research on stronger notions of robustness against strategic behavior in mechanism design.
Perhaps the closest notion of robustness (other than nonbossiness) to ours is \textit{obvious strategyproofness} by \citet{li2017obviously} (as well as variants thereof, such as strong obvious strategyproofness \citep{pycia2018obvious}), which, roughly speaking, requires that the worst thing that may happen to an agent (over other agents' actions) under truthful reporting must be at least as good as the best thing that may happen (again, over other agents' actions) under any possible deviation.
This is similar to our notion of nonbossiness at a high level, in the sense that both notions put limitations on how an agent may influence the utility of other agents.
However, as we show in Appendix~\ref{app:osp}, these two notions of robustness are not directly comparable, and in particular, neither notion subsumes the other.
Another notion of robustness is credibility by \citet{akbarpour2020credible}, which roughly says that the principal has no incentive to cheat in the mechanism.
In contrast, we assume full commitment power for the principal, and our notion of robustness concerns agents' incentives only.
Another related notion is \textit{group strategyproofness}, which tries to characterize when groups of agents can collude to increase each of their utilities. Group strategyproofness has been shown to be equivalent to nonbossiness plus incentive compatibility in certain settings~\citep{papai2000strategyproof,alva2017manipulation}, but one key difference between this line of work and our results is that it focuses on choice functions without money, while payments are a crucial component in our model. \citet{goldberg2005collusion} study a stronger variant of group strategyproofness (that they call ``$t$-truthfulness'') in single-parameter settings; however their notion is considerably stronger than nonbossiness, and the only mechanisms satisfying even $2$-truthfulness are ordinary sequential posted price mechanisms (in contrast to the expanded class of decision list nonbossy mechanisms we uncover).

\paragraph{Mechanisms with (non-vanishing) externalities.} The concept of ``externalities'' is well established in the economics literature, and there is a significant line of existing work on mechanism design with externalities \citep{jehiel1996not, jehiel1999multidimensional, segal1999contracting, bernstein2012contracting, bartling2016externality}, in addition to closely related lines of work on mechanism design with altruistic \cite{levine1998modeling, fehr2006economics} or spiteful agents \citep{brandt2001antisocial, morgan2003spite, brandt2005spiteful, zhou2007vindictive}. As far as we are aware, we are the first to specifically study the setting of externalities that tend to zero in size (or mechanisms robust to small externalities), and the techniques used for larger externalities are significantly different than the techniques we employ. 

\paragraph{Simple mechanisms and posted-price mechanisms.}
Another related topic is simple mechanisms, and in particular, posted-price mechanisms.
The study of simple mechanisms is in part driven by the fact that in rich environments (e.g., with multiple items and possibly non-additive valuations), optimal mechanisms in terms of revenue or welfare can be complex~\citep{thanassoulis2004haggling,manelli2007multidimensional,hart2015maximal,daskalakis2017strong} and/or hard to compute~\citep{mirrokni2008tight,daskalakis2014complexity}.
On the other hand, often there are relatively simple mechanisms, such as posted-price mechanisms, that can achieve good (often constant) approximations to the respective benchmarks.
Early research on the topic includes the work by \citet{hagerty1987robust}, who characterize posted-price mechanisms as essentially the only mechanisms such that each agent has a dominant strategy in the context of bilateral trade.
Along this line, \citet{makowski1994bayesian} characterize two classes of mechanisms defined by certain properties in a fairly general model, and show that both characterizations have simple interpretations;
\citet{barbera1995strategy} study strategyproof social choice functions in a setting where multiple divisible items are allocated to multiple agents, and show such social choice functions take the simple form of ``trading according to pre-specified proportions'';
\citet{satterthwaite2002optimality} show that a simple mechanism, the $k$-double auction, is asymptotically optimal in the worst case in markets with multiple buyers and sellers trading indivisible homogeneous items.
In terms of revenue, it is known that when agents are single-parameter or unit-demand, sequential posted-price auctions achieve a constant fraction of the optimal revenue~\citep{chawla2007algorithmic,chawla2010multi,kleinberg2012matroid,cai2019duality}.
Similar results have been established for additive agents~\citep{hart2017approximate,li2013revenue,babaioff2014simple,yao2015n,cai2019duality} and even subadditive agents~\citep{rubinstein2018simple,cai2017simple} using slightly less simple mechanisms.
As for welfare, it is known that anonymous item-pricing mechanisms achieve an approximation factor of $2$ for submodular/XOS agents~\citep{feldman2014combinatorial,dutting2020prophet}, and $O(\log \log m)$ for subadditive agents~\citep{dutting2020prophet,dutting2020log}.
Our results are most closely related to prior work on posted-price mechanisms: restricted to single-parameter settings, our results show that the class of mechanisms that are incentive compatible, individually rational and nonbossy generalize sequential posted-price mechanisms, in that incentive-compatible, individually rational and nonbossy mechanisms (corresponding to decision lists with exceptions) are slightly less simple, but strictly more powerful.
Clock auctions (see, e.g., \citep{christodoulou2022optimal,balkanski2022deterministic,feldman2022bayesian}) have also received significant attention, especially in the context of spectrum auctions.
Our results for single-parameter settings are related to clock auctions in the sense that both classes of mechanisms generalize sequential posted-price mechanisms (albeit in different directions) and preserve certain desirable properties.

\section{Preliminaries}

\subsection{Valuation Spaces, Outcome Spaces, and Mechanisms.}

We consider multiagent settings with $n$ agents.
Each agent $i$ is associated with a valuation space $\cV_i$ and a personal outcome space $\cO_i$.
Each valuation $v_i \in \cV_i$ consists of two components: (1) the primary valuation function (abusing notation, we also denote this component by $v_i$) maps every personal outcome $o_i \in \cO_i$ to a real number $v_i(o_i)$, which is the value $i$ derives from outcome $o_i$, and (2) the tiebreaking component $\vartriangleright_i$ specifies a total ordering over $\cO_i$, such that whenever two outcome-payment pairs induce the same quasi-linear utility (to be formally discussed later when primary incentive compatibility is defined), the agent always prefers the one with a ``larger'' outcome.\footnote{
    Most of our results hold ``almost everywhere'' when ties are broken by the mechanism.
    See Appendix~\ref{app:tiebreaking} for more details.
}
Abusing notation, we sometimes write $v_i = (v_i, \vartriangleright_i)$.
We also define the primary component of $\cV_i$ to be the collection of the primary components of all valuations $v_i \in \cV_i$, and the tiebreaking component similarly.
Throughout the paper, we focus on valuation spaces $\cV_i$ whose tiebreaking component is the set of all total orderings over $\cO_i$.
The joint valuation space $\cV$ specifies all possible combinations of valuations, and the joint outcome space $\cO \subseteq \prod_i \cO_i$ specifies all feasible combinations of personal outcomes.
For brevity and readability, we will omit the existence of the tiebreaking component whenever appropriate, and explicitly refer to it only when necessary.

\begin{example}
    Consider a single-parameter auction setting, where there are $k < n$ identical items for sale, and each agent wants at most one of them.
    Each $\cO_i$ can be either $\{0, 1\}$ (corresponding to deterministic allocations) or $[0, 1]$ (corresponding to randomized allocations).
    Each valuation space $\cV_i$ is isomorphic to the product space of (1) the set of nonnegative real numbers $\bR_+$ and (2) the set of all total orderings over $\cO_i$.
    The former corresponds to the primary component and the latter the tiebreaking component.
    This is because (assuming agents care about expected utility) each valuation $v_i \in \cV_i$ can be described by a nonnegative real number --- the value $v_i(1)$ of agent $i$ receiving an item --- and how ties are broken (again, this will be formally discussed soon).
    Suppose we consider deterministic mechanisms.
    Then the joint outcome space is defined to be
    \[
        \cO = \left\{o \in \{0, 1\}^n \,\middle|\, \sum_{i \in [n]} o_i \le k\right\}.
    \]
    This captures the fact that no more than $k$ agents can receive an item.
\end{example}

We focus on direct mechanisms in this paper.
A direct mechanism (or simply a mechanism for short) $M = (f, p)$ maps each combination of valuations $v = (v_1, \dots, v_n) \in \cV$ to a joint outcome $f(v) \in \cO$, and charges payments $p(v) \in \bR_+^n$.
Let $f_i(v)$ denote the personal outcome for agent $i$ in $f(v)$, and $f_{-i}(v)$ all other personal outcomes.
Similarly define $p_i$, $p_{-i}$.
We also generally use $o \in \cO$ and $q \in \bR_+^n$ to denote specific outcome/payment vectors, as opposed to the mappings $f: \cV \to \cO$ and $p: \cV \to \bR_+^n$.
For each vector of valuations $v \in \cV$ and each agent $i$, we use $v_{-i}$ to denote the valuations of all agents except $i$.
We similarly define $o_{-i}$, $q_{-i}$, $\cV_{-i}$, etc.

Note that under our definition, a direct mechanism elicits not only the primary valuation function, but also the tiebreaking ordering.
This is natural, and in some sense necessary in the context of our discussion.
To see why this is the case, consider a single-agent posted-price mechanism, where the price offered by the mechanism is $p$.
Further suppose the agent's primary valuation is $v = p$, which is reported to the mechanism, but the tiebreaking component remains private.
Then, it is impossible for the mechanism to decide whether or not to allocate the item to the agent while respecting the agent's tiebreaking ordering in all circumstances (which is necessary for primary incentive compatibility to be defined later), simply because the latter is unknown to the mechanism.
This is not an issue in the indirect version of the same mechanism, where the agent decides whether to buy the item based on both the primary valuation and the tiebreaking ordering, both known privately to the agent.
In other words, we must allow direct mechanisms to elicit both components of each agent's valuation, so our definition allows posted-price mechanisms, as well as other natural mechanisms, to admit direct implementations.
This does not necessarily mean we must directly elicit all this information {\em in reality}; rather, the indirect versions of the same mechanisms are likely better suited for practical purposes.

\subsection{Primary Incentive Compatibility and Individual Rationality}

The classical notions of (dominant-strategy) individual rationality and primary incentive compability are captured by the following definitions in our model.

\begin{definition}[Individual Rationality (IR)]
A mechanism $M = (f, p)$ is individually rational (IR) if for all $i$ and $v \in \cV$, 
\[
    v_i(f_i(v)) - p_i(v) \ge 0.
\]
\end{definition}

\begin{definition}[Primary Incentive Compatibility (PIC)]
A mechanism $M = (f, p)$ is primarily incentive compatible (PIC) if for all $i \in [n]$, $v = (v_i, v_{-i}) \in \cV$, and $v_i' \in \cV_i$,
\begin{align*}
    & v_i(f_i(v_i, v_{-i})) - p_i(v_i, v_{-i}) > v_i(f_i(v_i', v_{-i})) - p_i(v_i', v_{-i}), \text{ or} \\
    & v_i(f_i(v_i, v_{-i})) - p_i(v_i, v_{-i}) = v_i(f_i(v_i', v_{-i})) - p_i(v_i', v_{-i}) \text{ and } f_i(v_i, v_{-i}) \trianglerighteq_i f_i(v_i', v_{-i}).
\end{align*}
Here, $\vartriangleright_i$ is the tiebreaking component of $v_i$.
\end{definition}

In the next subsection, we will introduce a richer model of preferences incorporating secondary goals, which induces a stronger notion of incentive compatibility than PIC defined above.

\subsection{Incentives in the Presence of Hidden Secondary Goals}\label{sec:prelim_ext}

We now formally define agents' incentives in the presence of hidden secondary goals.
For each agent $i$, let $\succ_i$ be $i$'s preference over all possible combinations of outcomes and payments, i.e., over $\cO \times \bR_+^n$.
This preference is determined by $v_i \in \cV_i$ and secondary goals together.
We assume that $\succ_i$ can be decomposed lexicographically into:
\begin{itemize}
    \item An internal component over internal outcome-payment pairs $\cO_i \times \bR$ (denoted by $\succ_i^\pri$ for brevity) induced by the valuation $v_i$ in the straightforward way, such that internal outcome-payment pairs $(o_i, q_i)$ giving higher quasi-linear utility $v_i(o_i) - q_i$ are preferred, and ties are broken according to the tiebreaking ordering $\vartriangleright_i$.\footnote{%
        This allows one to rephrase the definition of PIC in the following way: a mechanism $M = (f, p)$ is PIC, if for all $i$, $v$ and $v_i'$,
            $(f_i(v_i, v_{-i}), p_i(v_i, v_{-i})) \succeq_i^\pri (f_i(v_i', v_{-i}), p_i(v_i', v_{-i}))$.
    }%
    Here, we assume the tiebreaking ordering $\vartriangleright_i$ is strict, which ensures that $\succ_i^\pri$ is a linear order over the continuum space $\cO_i \times \bR$.  The case with possible indifferences is treated in Appendix~\ref{app:tiebreaking}.
    \item An external component (capturing secondary goals) $\succ_i^\pub$ over $\cO_{-i} \times \bR_+^{n - 1}$.
\end{itemize}
Somewhat abusing notation, we also write $\mathbin{\succ_i} = (v_i, \mathbin{\succ_i^\pub})$.
For any two combinations of outcomes and payments $(o, q)$ and $(o', q') \in \cO \times \bR_+^n$, $(o, q) \succeq_i (o', q')$ iff: $(o_i, q_i) \succ_i^\pri (o_i', q_i')$, or $(o_i, q_i) = (o_i', q_i')$ and $(o_{-i}, q_{-i}) \succeq_i^\pub (o_{-i}', q_{-i}')$.
That is, the valuation decides the agent's preference between two combinations, unless the internal outcome-payment pairs are exactly the same, in which case the external component decides the preference.
Conceptually, the external component captures secondary goals, and its secondary nature is captured by the fact that it only matters when the agent is indifferent with respect to the internal component.

We place no restrictions on the external component of the preference $\succ_i^{\pub}$: we allow each $\succ_i^{\pub}$ to be any ordering over the outcomes $\cO_{-i} \times \bR_+^{n - 1}$ of other agents. That is, we wish to design mechanisms that are robust to any possible secondary goal that an agent might have.\footnote{In Appendix~\ref{app:restricted_externalities} we show this assumption can in fact be significantly relaxed.}
In particular, a sensible $\succ_i^\pub$ need not be induced by a utility function with nice properties, e.g., continuity in other agents' payments.
Consider the following example: 
\begin{example}
    Two agents compete for one item in a second-price auction.
    Both agents' primary goal is to maximize quasi-linear utility, while their secondary goal is to avoid forcing each other to make ``inconvenient'' payments.
    That is, the external component $\succ_i^\pub$ of agent $i$'s preference prefers outcome-payment pairs $(o_{-i}, q_{-i})$ where $q_{-i}$ is an integral multiple of $0.01$, to all other outcome-payment pairs.
    Intuitively and informally, when one agent expects the other agent to have a higher value, they would round their own value to, say, the closest multiple of $0.01$ when placing the bid, so as to achieve their secondary goal.
    Note that if we were to capture the above secondary goal using a utility function, it would necessarily ``jump'' at each integral multiple of $0.01$.
\end{example}

In fact, the external component $\succ_i^\pub$ may not correspond to a utility function at all.
Consider the following example:
\begin{example}
    There are $3$ agents, and we focus on agent $1$, whose secondary goal itself is lexicographical: agent $1$ prefers higher payments made by agent $2$; fixing the payment made by agent $2$, agent $1$ prefers higher payments made by agent $3$.
    Formally, $\succ_1^\pub(o_{-1}, q_{-1})$ depends only on $q_2$ and $q_3$, and $(q_2, q_3) \succ_1^\pub (q_2', q_3')$ iff (1) $q_2 > q_2'$ or (2) $q_2 = q_2'$ and $q_3 > q_3'$.
    Intuitively, agent $1$'s secondary goal is essentially for agent $2$ (the major competitor) to suffer, but it would be even better if agent $3$ (the minor competitor) suffers at the same time.
    It is known that such a preference cannot be captured by any utility function over $q_2$ and $q_3$.
\end{example}

We now describe the behavior of agents.
Each agent $i$, given $i$'s valuation $v_i$, all other agents' valuations $v_{-i}$, and the external component of $i$'s preference $\succ_i^\pub$, reports a possibly nontruthful valuation $v_i' \in \cV_i$ to the mechanism, where the goal is to achieve a most preferable combination of outcomes and payments according to $\succ_i$.
As in environments without secondary goals, the exact behavior of agents also depends on what they know about each other.
In the rest of this paper, we focus on dominant-strategy incentive compatibility, which means an agent never gets more desirable combinations of outcomes and payments by misreporting their private valuation, no matter what other agents do (in particular, without loss of generality we can assume agents know the valuations and preferences of all other agents).
This is captured by the following definition.

\begin{definition}[Incentive Compatibility (IC)]
    A mechanism $M = (f, p)$ is incentive compatible, if for any agent $i$, valuation vector $v \in \cV$, public component of $i$'s preference $\mathbin{\succ_i^\pub}$, and possible deviation $v_i' \in \cV_i$,
    \[
        (f(v_i, v_{-i}), p(v_i, v_{-i})) \succeq_i (f(v_i', v_{-i}), p(v_i', v_{-i})),
    \]
    where $\mathbin{\succ_i} = (v_i, \mathbin{\succ_i^\pub})$.
\end{definition}

Note that IC as defined above is a stronger notion of robustness than PIC, since for any two combinations of outcomes and payments $(o, q)$ and $(o', q')$,
\[
    (o, q) \succeq_i (o', q') \implies (o_i, q_i) \succeq_i^\pri (o_i', q_i').
\]

\begin{proposition}
    Any IC mechanism $M$ is also PIC.
\end{proposition}

\section{Nonbossy Mechanisms}

In this section, we introduce the notion of nonbossy mechanisms, and show that it tightly characterizes the class of mechanisms that are robust with secondary goals.

\begin{definition}[Nonbossiness (NB)]
A mechanism $M = (f, p)$ is nonbossy (NB) if for all $i$, $v = (v_i, v_{-i}) \in \cV$ and $v_i' \in \cV_i$,
\begin{align*}
    \text{if:}\quad & f_i(v_i, v_{-i}) = f_i(v_i', v_{-i}) \text{ and } p_i(v_i, v_{-i}) = p_i(v_i', v_{-i}) \\
    \text{then:}\quad & f(v_i, v_{-i}) = f(v_i', v_{-i}) \text{ and } p(v_i, v_{-i}) = p(v_i', v_{-i}).
\end{align*}
\end{definition}

In words, the above definition says that a mechanism is nonbossy if no agent can change another agent's personal outcome or payment without changing their own personal outcome or payment.
Below we show that the semantics of nonbossiness extends well beyond the above definition --- in fact, restricted to PIC mechanisms, the family of nonbossy mechanisms is precisely the family of IC mechanisms in the presence of secondary goals.

\begin{theorem}[Semantics of Nonbossiness]
\label{thm:semantic}
    A mechanism is IC if and only if it is PIC and NB.
\end{theorem}
\begin{proof}[Proof Sketch]\footnote{For the sake of brevity, we defer most full proofs to Appendix \ref{sec:omitted}. Where instructive, we supply proof sketches in their place.}
If a mechanism is nonbossy, it is impossible for an agent's deviation to modify the external component of their preference without also modifying the internal component of their preference. Since in PIC mechanisms it is impossible for an agent to deviate and improve their internal preference, this means that nonbossy PIC mechanisms are IC. 

On the other hand, if a mechanism is not nonbossy, then there is a non-truthful deviation for some agent $i$ which modifies the outcomes of the other agents but not $i$. Such a mechanism cannot be IC, since it is possible that agent $i$ prefers this deviation in the external component of their preference.
\end{proof}

\paragraph{Deterministic vs.\ randomized mechanisms.}
We do not explicitly distinguish between deterministic and randomized mechanisms, as randomization can be captured by extending the outcome space to include all convex combinations of outcomes.
Nevertheless, we remark that nonbossiness is generally easy to achieve with randomization: for example, when each $\cV_i = \bR_+$ and each $\cO_i = [0, 1]$ (which correspond to single-parameter environments), a mechanism $(f, p)$ is nonbossy as long as each $f_i(v_i, v_{-i})$ is strictly monotone in $v_i$.
On the other hand, recall that Myerson's characterization \citep{myerson1981optimal} states that any PIC mechanism in single-parameter environments must be weakly monotone.
Therefore, one can easily adapt any PIC mechanism to be also nonbossy with negligible loss in the following way: run the original PIC mechanism with probability $1 - \eps$, and any strictly monotone PIC mechanism with probability $\eps$, where $\eps > 0$ is arbitrarily small.
In light of this, we focus on deterministic mechanisms in the rest of the paper: for example, we consider $\cO_i = \{0, 1\}$ rather than $[0, 1]$ in single-parameter environments.

\section{Payment Characterization for Nonbossy Mechanisms}
\label{sec:payment}

In this section, we study the basic structure of nonbossy mechanisms. Specifically, whenever an environment has an expressive enough class of valuation functions, we provide a strong characterization of the payment rules of nonbossy mechanisms in this environment. 
We mathematically formalize this notion of expressivity via what we call the ``upper semilattice property'', defined below.

\begin{definition}[Upper Semilattice Property]
\label{def:upper_semilattice}
    A pair $(\cV_i, \cO_i)$ has the upper semilattice property if for any $o_i \in \cO_i$, $v_i, v_i' \in \cV_i$, there exists $v_i'' \in \cV_i$ such that for all $o_i' \in \cO_i$,
    \[
        v_i''(o_i) - v_i''(o_i') \ge \max\{v_i(o_i) - v_i(o_i'), v_i'(o_i) - v_i'(o_i')\}.
    \]
    We say $v_i''$ is a common upper bound of $v_i$ and $v_i'$ with respect to $o_i$.
\end{definition}

Note that the above property is independent of the tiebreaking component of a valuation.
While the upper semilattice property may appear strong and/or counterintuitive, it turns out that most natural settings commonly studied in the context of mechanism design in fact have this property. For example, it is straightforward to check that in single-parameter settings (where a valuation $v_i \in \cV_{i}$ for a single agent is parameterized by an arbitrary nonnegative real number), Definition \ref{def:upper_semilattice} holds (e.g., if $o_i$ is allocation of the item and $o'_i$ is non-allocation, we can take $v''_i = \max(v_i, v'_i)$). In fact, the upper semilattice property holds for all the following settings:

\begin{itemize}
    \item
    {\bf Common outcome, complete domain:} $\cO_1 = \cO_2 = \dots = \cO_n$, and $\cO = \{(o_c, \dots, o_c) \mid o_c \in \cO_1 = \dots = \cO_n\}$.
    The primary component of each $\cV_i$ is the collection of all functions from $\cO_i$ to $\bR_+$.
    This corresponds to settings where the mechanism makes a single public decision that affects all agents.
    \item
    {\bf Single-parameter agents:} $\cO_1 = \cO_2 = \dots = \cO_n = \{0, 1\}$, and $\cO \subseteq \times_i \cO_i$.
    The primary componet of each $v_i \in \cV_i$ is described by a nonnegative number $x_i$, where $v_i(1) = x_i$ and $v_i(0) = 0$.
    This corresponds to settings where the mechanism allocates identical items to agents subject to an arbitrary feasibility constraint, where each agent is interested in at most $1$ item.
    \item
    {\bf Combinatorial auctions, single-minded agents:} there is a ground set $M$ of items, $\cO_1 = \cO_2 = \dots = \cO_n = 2^M$, and $\cO \subseteq \times_i \cO_i$.
    The primary component of each $v_i \in \cV_i$ is described by a subset $S_i$ of $M$ and a real number $x_i$, where $v_i(T) = x_i$ if $S_i \subseteq T$ and $v_i(T) = 0$ otherwise.
    This corresponds to settings where the mechanism allocates heterogeneous items to agents subject to an arbitrary feasibility constraint, where each agent is only interested in getting all items in a certain set.
    \item
    {\bf Combinatorial auctions, valuations ``between'' additive and XOS:} there is a ground set $M$ of items, $\cO_1 = \cO_2 = \dots = \cO_n = 2^M$, and $\cO \subseteq \times_i \cO_i$.
    The primary component of each $\cV_i$ is simultaneously a superset of all monotone additive valuations and a subset of all monotone XOS valuations.
    This corresponds to settings where the mechanism allocates heterogeneous items to agents subject to an arbitrary feasibility constraint, where possible valuations are rich enough to subsume all additive ones, but never go beyond XOS ones.
    \item
    {\bf Combinatorial auctions, valuations ``beyond'' subadditive:} there is a ground set $M$ of items, $\cO_1 = \cO_2 = \dots = \cO_n = 2^M$, and $\cO \subseteq \times_i \cO_i$.
    The primary component of each $\cV_i$ is simultaneously a superset of all subadditive valuations.
    This corresponds to settings where the mechanism allocates heterogeneous items to agents subject to an arbitrary feasibility constraint, where possible valuations are rich enough to subsume all monotone subadditive ones.
    \item
    {\bf Metric space:} there is a metric space $(X_i, d_i)$ for each agent $i$ where $O_i \subseteq X_i$.
    The primary component of each $v_i \in \cV_i$ is induced by a point $x_i \in X_i$, such that $v_i(o_i) = - d_i(x_i, o_i)$.
    This corresponds to settings where each agent is located in a metric space, which also contains all possible outcomes, and the agent wants the distance to the outcome to be as small as possible (e.g., in facility location games). 
\end{itemize}

\begin{proposition}
\label{prop:payment_applicability}
    All the above settings have the upper semilattice property, and therefore admit the payment characterization.
\end{proposition}

We now present our payment characterization for nonbossy mechanisms in such environments.
Our characterization states that whenever the environment has the upper semilattice property, for any PIC and NB mechanism, the payment rule can be written as a function of the outcome vector only.

\begin{theorem}[Payment Characterization]
\label{thm:payment}
    If for all $i \in [n]$, $(\cV_i, \cO_i)$ has the upper semilattice property, then for any PIC and NB mechanism $M = (f, p)$,
    \[
        p(v_1, \dots, v_n) = p(v_1', \dots, v_n') \quad \text{if} \quad f(v_1, \dots, v_n) = f(v_1', \dots, v_n').
    \]
\end{theorem}
\begin{proof}
    Fix any PIC and NB mechanism $(f, p)$.
    Consider any two valuation profiles $v_1, \dots, v_n$ and $v_1', \dots, v_n'$ where 
    \[
        f(v_1, \dots, v_n) = f(v_1', \dots, v_n') = (o_1, \dots, o_n).
    \]
    For each $i \in [n]$, let $v_i''$ be a common upper bound of $v_i$ and $v_i'$, satisfying for any $o_i' \in \cO_i \setminus \{o_i\}$,
    \[
        v_i''(o_i) - v_i''(o_i') \ge \max\{v_i(o_i) - v_i(o_i'), v_i'(o_i) - v_i'(o_i')\}.
    \]
    Moreover, we choose $v_i''$ such that its tiebreaking component orders $o_i$ before all other outcomes in $\cO_i$.
    We inductively argue that for any $i \in [n]$,
    \[
        f(v_1, \dots, v_n) = f(v_1'', \dots, v_i'', v_{i + 1}, \dots, v_n) \quad \text{and} \quad f(v_1', \dots, v_n') = f(v_1'', \dots, v_i'', v_{i + 1}', \dots, v_n'),
    \]
    and
    \[
        p(v_1, \dots, v_n) = p(v_1'', \dots, v_i'', v_{i + 1}, \dots, v_n) \quad \text{and} \quad p(v_1', \dots, v_n') = p(v_1'', \dots, v_i'', v_{i + 1}', \dots, v_n').
    \]
    In particular, this implies
    \[ p(v_1, \dots, v_n) = p(v_1'', \dots, v_n'') = p(v_1', \dots, v_n').  \]
    To show this, we only need to argue that if agent $i$ changes their reported valuation from $v_i$ to $v''_i$, then their outcome does not change, i.e.,
    \[
        f_i(v_1'', \dots, v_{i - 1}'', v_i, v_{i+1} \dots, v_n) = f_i(v_1'', \dots, v_{i-1}'', v_i'', v_{i + 1}, \dots, v_n).
    \]
    If this is the case, then PIC implies that agent $i$ must receive the same payment under both types, i.e.,
    \[
        p_i(v_1'', \dots, v_{i - 1}'', v_i, v_{i+1} \dots, v_n) = p_i(v_1'', \dots, v_{i-1}'', v_i'', v_{i + 1}, \dots, v_n),
    \]
    and NB implies that the payments and outcomes do not change for any other agent (which was the inductive claim to be proved).
    Consider the menu agent $i$ faces when the other agents' valuations are $(v_1'', \dots, v_{i - 1}'', v_{i + 1}, \dots, v_n)$.
    When $i$'s valuation function is $v_i$, $o_i$ is the utility-maximizing option.
    Moreover, by the construction of $v_i''$, when $i$'s valuation function is $v_i''$, $o_i$ can only become even more desirable compared to other options, so PIC ensures that
    \[
        f_i(v_1'', \dots, v_i'', v_{i + 1}, \dots, v_n) = o_i.
    \]
    The other part of the argument (for $v'$) is completely symmetric.
\end{proof}

The above characterization generalizes Theorem 2 in \citep{schummer2000eliciting}, which establishes the same property in the special case of allocating heterogeneous items to unit-demand agents (i.e., matching).
In fact, \citet{schummer2000eliciting} studies a model where the mechanism allocates both a set of indivisible items and a fixed amount of a divisible good (e.g., money), and the results (including Theorem 2) apply regardless of whether negative consumption of the divisible good is allowed.

\paragraph{Necessity of the upper semilattice property.}
One may naturally wonder if the upper semilattice property is necessary for the characterization --- below we show that it in fact is, even if we further restrict our attention to IR mechanisms.

\begin{proposition}
\label{prop:payment_counterexample}
    There exists a PIC, NB, and IR mechanism $M = (f, p)$ where the above payment characterization fails.
\end{proposition}
\begin{proof}
    Let $n = 2$, $\cV_1 = \{x_1, x_2\}$, $\cV_2 = \{y_1, y_2\}$, $\cO_1 = \cO_2 = \{o_1, o_2, o_3\}$, and $\cO = \{(o_1, o_1), (o_2, o_2), (o_3, o_3)\}$.
    Let the valuations be such that
    \[
        \begin{array}{cccc}
            & o_1 & o_2 & o_3 \\
            x_1 & 1 + \eps & 0 & 1 + 2 \eps \\
            x_2 & 1 + \eps & 1 & 2 \eps \\
            y_1 & 1 + \eps & 1 & 2 \eps \\
            y_2 & 1 + \eps & 0 & 1 + 2 \eps
        \end{array},
    \]
    where $\eps > 0$ is a small quantity (the only goal that $\eps$ serves is removing the need for tiebreaking).
    Consider the following mechanism.
    \begin{itemize}
        \item $f(x_1, y_1) = (o_1, o_1)$, $p(x_1, y_1) = (1, 0)$,
        \item $f(x_2, y_2) = (o_1, o_1)$, $p(x_2, y_2) = (0, 1)$,
        \item $f(x_2, y_1) = (o_2, o_2)$, $p(x_2, y_1) = (0, 0)$,
        \item $f(x_1, y_2) = (o_3, o_3)$, $p(x_1, y_2) = (0, 0)$.
    \end{itemize}
    This does not satisfy the payment characterization, because when the outcome vector is $(o_1, o_1)$, the payment vector can either be $(1, 0)$ or $(0, 1)$.
    On the other hand, observe that the mechanism is NB because the (common) outcome changes whenever either agent's valuation changes, and verifying that the mechanism is PIC simply involves checking all relevant cases: when agent $2$ reports $y_1$, agent $1$ would not misreport $x_1$ as $x_2$, because that would decrease their utility from $\eps$ to $0$.
    Agent $1$ also would not misreport $x_2$ as $x_1$, because that would decrease their utility from $1$ to $\eps$.
    When agent $2$ reports $y_2$, agent $1$ would not misreport $x_1$ as $x_2$, because that would decrease their utility from $1 + 2\eps$ to $1 + \eps$.
    Agent $1$ also would not misreport $x_2$ as $x_1$, because that would decrease their utility from $1 + \eps$ to $2\eps$.
    The case of agent $2$ is symmetric.
\end{proof}

\section{Full Characterizations in Specific Settings}

In this section, we present full characterizations of PIC, NB, and possibly IR mechanisms in several specific settings.

\subsection{Gibbard-Satterthwaite Style Characterization for the Common Outcome Setting}

First we present a characterization for the setting where all agents share a common outcome, as described in Section~\ref{sec:payment}.
Somewhat surprisingly, this characterization shows that the seemingly richer space of PIC and NB mechanisms (which, notably, involve payments) contains only ``trivial mechanisms''. These mechanisms must have the same form as the characterization of non-manipulable social choice rules in the celebrated result of Gibbard and Satterthwaite \citep{gibbard1973manipulation,satterthwaite1975strategy}: they are either dictatorships, or only two outcomes are relevant.

\begin{theorem}\label{thm:common}
    Consider the setting where all agents share a common outcome and valuations are from the complete domain, i.e., $\cO = \{(o_c, \dots, o_c) \mid o_c \in \cO_1 = \dots = \cO_n\}$ where $|\cO_i| < \infty$, and each $\cV_i$ is the collection of all functions from $\cO_i$ to $\bR_+$.
    In this setting, any PIC and NB mechanism satisfies the Gibbard-Satterthwaite characterization, i.e., either it is dictatorship or it uses only two outcomes.
    Formally, any PIC and NB mechanism $M = (f, p)$ in this setting must satisfy at least one of the following two conditions:
    \begin{itemize}
        \item Dictatorship: there exists an agent $i$ such that $(f(v), p(v)) = (f(v_i), p(v_i))$, i.e., $f$ and $p$ depend only on agent $i$'s valuation.
        \item Two outcomes: the image of $(f, p)$ has cardinality at most $2$, i.e., $|\{(f(v), p(v)) \mid v \in \cV\}| \le 2$.
    \end{itemize}
\end{theorem}
\begin{proof}[Proof Sketch]
If the mechanism $M$ is NB, then by our payment characterization, the payment $p$ is solely a function of the eventual outcome $o$. Consider the social choice problem (where agents wish to elect a single outcome) where the preferences of agent $i$ over the $|\cO|$ outcomes are given by the values of $v_{i}(o) - p_{i}(o)$ in increasing order. Then $M$ is a mechanism that solves this social choice problem; moreover, we can show that if $M$ is PIC for the original problem, it must be PIC for this social choice problem, and that every possible collection of $n$ orderings over the outcomes can occur. It follows from the Gibbard-Sattherthwaite theorem that $M$ must take the above form.
\end{proof}

See Appendix~\ref{sec:omitted} for the full proof of the theorem.
Given the equivalence between IC and nonbossiness (Theorem~\ref{thm:semantic}), the above result can be interpreted in the following way: when choosing a common outcome, only trivial mechanisms can be robust against strategic behavior if agents care (even infinitesimally) about each other's utilities.

\subsection{Characterizations for Single-Parameter Environments}

We now turn to single-parameter environments, where each $\cO_i = \{0, 1\}$, and the primary component of each $v_i \in \cV_i$ satisfies $v_i(1) \ge v_i(0) = 0$.

\paragraph{Notational shorthand.}
For ease of presentation, in single-parameter environments, we use the following notation that incorporates both components of a valuation: for a nonnegative real number $x \in \bR_+$, we say $v_i = x_+$ if the primary component satisfies $v_i(1) = x$ and the tiebreaking component satisfies $1 \vartriangleright_i 0$; similarly, we say $v_i = x_-$ if $0 \vartriangleright_i 1$.
We also write $\cV_i = \eR = \bR_+ \times \{\mathbin{+}, \mathbin{-}\}$, i.e., $\cV_i$ is isomorphic to the set of ``extended'' nonnegative real numbers.
Note that $\eR$ is naturally ordered: for each $x \in \bR_+$, we say $x_+ > x_-$, and for each $x, y \in \bR_+$ where $x < y$, we say $x_+ < y_-$.
Similarly, we also allow comparisons between extended real numbers and standard real numbers: for each $x \in \bR_+$, we say $x_- < x < x_+$.
Thus, $\eR \cup \bR_+$ is equipped with a natural linear order, which will be useful throughout the rest of the paper.

\paragraph{The single-item characterization.}
Our first characterization applies to the single-item setting, where $\cO = \left\{o \in \{0, 1\}^n \,\middle|\, \sum_{i \in [n]} o_i \le 1\right\}$.
The characterization involves (the direct versions) of sequential posted-price mechanisms.
A sequential posted-price mechanism approaches each agent at most once in a predetermined order, makes a (possibly personalized) take-it-or-leave-it offer to each agent visited, and allocates the item to the first agent who accepts the offer.
The direct implementation of such a mechanism is defined by a permutation $\sigma: [n] \to [n]$ of all agents, together with a price $p_i \in \bR_+ \cup \{\infty\}$ for each agent $i \in [n]$.
The mechanism examines each agent $\sigma(j)$ for $j \in [n]$ and allocates the item to the first agent in the permutation, i.e., $\sigma(j)$ with the smallest $j$, where $v_{\sigma(j)} \ge p_{\sigma(j)}$.
If no such agent exists, then the mechanism makes no allocation.

\begin{theorem}
\label{thm:single-item}
    In the single-item setting, any PIC, IR, and NB mechanism is a sequential posted-price mechanism.
\end{theorem}

It is easy to see that any sequential posted-price mechanism is PIC, IR, and NB.
Conceptually, the above characterization says that sequential posted-price mechanisms are the only format of single-item auctions that are robust in the presence of secondary goals.

\paragraph{The single-parameter characterization via decision lists.}
Theorem \ref{thm:single-item} is a corollary of our characterization of nonbossy mechanisms in general single-parameter environments, which states that in the single-parameter setting, any PIC, IR, and NB mechanism $M = (f, p)$ is a ``decision list'' of the following form: each feasible outcome $o \in \cO$ is associated with a payment vector $p(o)$, where $p_i(o) = 0$ if $o_i = 0$ for each $i$, and a set $E(o) \subseteq \cO$ of feasible outcomes (we say $E(o)$ is the exception list of $o$).
We say a valuation vector $v$ satisfies an outcome $o$ if for all $i$ where $o_i = 1$, $v_i \ge p_i(o)$.\footnote{
    Note that here we are comparing $v_i \in \eR$ with $p_i(o) \in \bR_+$.
    The comparison is based on the natural linear order over $\eR \cup \bR_+$ discussed above.
    Intuitively, $v$ satisfying $o$ means when the agents' valuations are $v \in \eR^n$, each agent $i$ where $o_i = 1$ is willing to pay $p_i(o) \in \bR_+$ for an item, taking into consideration how they break ties.
}
We choose a linear order among all feasible outcomes, and the mechanism checks these outcomes in that order, and outputs the first outcome $o$ satisfied by $v$ such that no feasible outcome $o' \in E(o)$ is also satisfied.
This is formally captured by the following claim.

\begin{theorem}
\label{thm:single-parameter}
    In the single-parameter setting, for any PIC, IR and NB mechanism $M = (f, p)$, there exists a linear order of outcomes $o^{(1)}, \dots, o^{(k)}$ where $k = |\cO|$, such that for all $v \in \cV$, $f(v) = o^{(j)}$ and $p(v) = p(o^{(j)})$, where
    \[
        j = \min\{j' \in [k]: v \text{ satisfies } o^{(j')} \wedge (\forall o \in E(o^{(j')}),\, v \text{ does not satisfy } o)\}.
    \]
    Moreover, there exists a way to choose the exception lists such that the order does not matter, i.e., for any valuation vector $v$, there is precisely one $j$ such that
    \[
        v \text{ satisfies } o^{(j)} \wedge (\forall o \in E(o^{(j)}),\, v \text{ does not satisfy } o).
    \]
\end{theorem}

Unlike in the single-item case, the single-parameter characterization is considerably more general than sequential posted-price mechanisms.
In particular, PIC, IR, and NB mechanisms allow posting a vector of prices to all agents simultaneously, which is accepted if and only if each agent accepts their personal component of the price vector.
We provide an overview of proof techniques used in our single-item and single-parameter characterizations in Appendix~\ref{app:single-parameter_overview}.

\paragraph{Implications of the single-parameter characterization.}
Although the full statement of Theorem~\ref{thm:single-parameter} is quite lengthy, it does have useful and neat implications.
For example, Theorem~\ref{thm:single-parameter} implies that the space of mechanisms that are PIC, IR, and NB is finite-dimensional.
On the negative side, this immediately rules out several popular auction formats such as first-price auctions and second-price auctions.
On the positive side, since the design space is finite-dimensional, in principle, one can search for the optimal mechanism by brute force.
The interpretation of sequentially posting vectors of prices also provides a more concrete and operable way to design mechanisms that are PIC, IR and NB.

\paragraph{Condition for order-oblivious decision lists to be PIC, IR and NB.}
In Theorem~\ref{thm:single-parameter}, we prove a necessary condition for a mechanism to be PIC, IR, and NB: it must take the form of an order-oblivious decision list with exceptions.
However, not all mechanisms induced by order-oblivious decision lists with exceptions are PIC.\footnote{Though, as we show in the proof of Lemma~\ref{lem:order-oblivious_ic}, they are always IR and NB.}
Below we give a necessary and sufficient condition for an order-oblivious decision list with exceptions to be PIC, IR, and NB, which, together with Theorem~\ref{thm:single-parameter}, constitutes an if-and-only-if characterization of PIC, IR and NB mechanisms in single-parameter environments, i.e., a mechanism is PIC, IR and NB, if and only if it can be written as an order-oblivious decision list with exceptions satisfying the condition in the following claim (proof defered to Appendix~\ref{sec:omitted}).

\begin{lemma}
\label{lem:order-oblivious_ic}
    Consider any single-parameter setting defined by $\cO \subseteq \{0, 1\}^n$.
    Fix any order-oblivious decision list defined by exception lists $\{E(o)\}_{o \in \cO}$ and associated price vectors $\{p(o)\}_{o \in \cO}$ which induces a well-defined mechanism.
    For each vector $u \in \eR^n$, let
    \[
        O^u = \{o \in \cO \mid \forall i,\, u_i \ge p_i(o)\}.
    \]
    Then, the order-oblivious decision list induces a PIC, IR and NB mechanism if and only if it satisfies the following condition: for each agent $i$ and vector $u \in \eR^n$, there exists at most one $o \in O^u$ satisfying:
    \[
        o_i = 1 \quad \text{and} \quad \{o' \in E(o) \cap O^u \mid p_i(o') \le p_i(o)\} = \emptyset.
    \]
\end{lemma}

Now we can put Theorem~\ref{thm:single-parameter} and Lemma~\ref{lem:order-oblivious_ic} together and obtain the following claim.

\begin{corollary}
\label{cor:single-parameter-iff}
    In any single-parameter setting, a mechanism is PIC, IR and NB if and only if it can be written as an order-oblivious decision list (as defined in Theorem~\ref{thm:single-parameter}) satisfying the condition in Lemma~\ref{lem:order-oblivious_ic}.
\end{corollary}

\newpage

\appendix

\section{Details Regarding Tiebreaking}\label{app:tiebreaking}

In this appendix, we discuss how tiebreaking affects our results.
Our model treats tiebreaking as part of an agent's preferences, i.e., the internal component of each agent's preferences induces a strict total order over all (outcome, payment) pairs.
Alternatively, one may also consider the model where agents are indifferent between two (outcome, payment) pairs whenever they induce the same quasi-linear utility.
In this section, we discuss the counterparts of our main results in this alternative model.
Technically, we will try to avoid repetition and focus on the essentially different parts of the results and proofs.

\paragraph{An alternative notion of IC.}
With the tiebreaking component $\mathbin{\vartriangleright_i}$ removed, an agent's internal preferences $\mathbin{\succ_i^\pri}$ is solely determined by the primary component of their valuation $v_i: \cO \to \bR$, i.e., $(o_i, q_i) \succ_i^\pri (o_i', q_i') \iff v_i(o_i) - q_i > v_i(o_i') - q_i'$.
Note that this is a weak order with equivalences, which induces an alternative notion of PIC (henceforth \altIC).
Plugging this into the definition of IC, we obtain an alternative version of IC with secondary goals.
The corresponding alternative notion of NB is the following:

\begin{definition}[\altNB]
a mechanism $M = (f, p)$ is \altNB, if for all $i$, $(v_i, v_{-i}) \in \cV$, and $v_i' \in \cV_i$,
\begin{align*}
    \text{if:} \quad & v_i(f_i(v_i, v_{-i})) - p_i(v_i, v_{-i}) = v_i(f_i(v_i', v_{-i})) - p_i(v_i', v_{-i}) \\
    \text{then:} \quad & v_{i'}(f_{i'}(v_i, v_{-i})) - p_{i'}(v_i, v_{-i}) = v_{i'}(f_{i'}(v_i', v_{-i})) - p_{i'}(v_i', v_{-i}),\, \forall i'.
\end{align*}
\end{definition}

\paragraph{Payment characterization.}
We then proceed to the technical results.
We begin with the payment characterization, which now requires a stronger upper semilattice property.

\begin{definition}[Strong Upper Semilattice Property]
    A pair $(\cV_i, \cO_i)$ has the upper semilattice property if for any $o_i \in \cO_i$, $v_i, v_i' \in \cV_i$, there exists $v_i'' \in \cV_i$ such that for all $o_i' \in \cO_i$,
    \[
        v_i''(o_i) - v_i''(o_i') > \max \{v_i(o_i) - v_i(o_i'), v_i'(o_i) - v_i'(o_i')\}.
    \]
    We say $v_i''$ is a common upper bound of $v_i$ and $v_i'$ with respect to $o_i$.
\end{definition}

While the above property is stronger (with $\ge$ in the original definition replaced by $>$), it is still satisfied by (among others) the two settings of particular interest:
\begin{itemize}
    \item Common outcome, complete domain: where each $\cV_i$ is the collection of all functions from $\cO_i$ to $\bR_+$.
    The proof is straightforward.
    \item Single-parameter agents {\em with positive values}: where each $\cO_i = \{0, 1\}$ and each $\cV_i$ is isometric to $\{x \in \bR \mid x > 0\}$.
    Note that here we require each agent's valuation of receiving an item to be strictly positive.
    This ensures that the strong upper semilattice holds even if the outcome $o_i \in \cO_i$ of interest is $o_i = 0$.
    In such cases, we need to choose $v_i''$ to be strictly smaller than both $v_i$ and $v_i'$.
    The proof is again straightforward except for the above change.
\end{itemize}

The payment characterization under \altNB{} then requires the above strong upper semilattice property.
\begin{theorem}[Payment Characterization under \altNB]
    If for all $i \in [n]$, $(\cV_i, \cO_i)$ has the strong upper semilattice property, then for any \altIC{} and \altNB{} mechanism $M = (f, p)$,
\end{theorem}
The proof of the claim is essentially the same as that of Theorem~\ref{thm:payment}, except that here we do not need to deal with ties.

\paragraph{Characterization for the common outcome setting.}
For the common outcome setting, our claim in Theorem~\ref{thm:common} holds for \altNB{} mechanisms without any change.
The only notable difference in the proof is that now we reduce to the corresponding social choice problem with weak preferences --- indeed, the Gibbard-Satterthwaite theorem holds for weak preferences too (see, e.g., \citep{schmeidler1978two}).

\paragraph{Characterization for the single-parameter setting.}
Now we proceed to the single-parameter setting.
In short, under the alternative definitions, our results hold almost everywhere, except on the part of the joint valuation space that lead to ties or involve valuations of $0$.
This is not uncommon in similar characterizations, e.g., the main results in the study on credible mechanisms \citep{akbarpour2020credible} exhibit a similar flavor.
Below we discuss the modifications needed to establish our results almost everywhere under \altNB.

Here we work with strictly positive values, where the payment characterization under \altNB{} applies.
In other words, what we establish here do not apply to the part of the joint valuation space where some agents have valuations of $0$.
We first introduce the following weaker version of Lemma~\ref{lem:consistency}.

\begin{lemma}
\label{lem:weak_consistency}
    Consider any \altIC, IR, and \altNB{} mechanism $M = (f, p)$ in a single-parameter environment.
    Fix two different feasible outcomes $o$ and $o'$, and let $p(o)$, $p(o')$ be the corresponding payment vectors.
    Choose any two valuation vectors where both $v$ and $v'$ {\em strictly} satisfy both $o$ and $o'$ (i.e., $\min\{v_i, v_i'\} > \max\{p_i(o), p_i(o')\}$ for each $i \in [n]$). 
    Then,
    \[
        f(v) = o \implies f(v') \ne o'.
    \]
\end{lemma}

The proof is essentially the same as that of Lemma~\ref{lem:consistency}, except that we replace $\ge$ with $>$ to avoid ties.
From here on, we say a valuation vector $v$ strictly satisfies an outcome $o$ if $v_i > p_i(o)$ for each $i$.
Lemma~\ref{lem:weak_consistency} allows us to establish an almost-everywhere version of Theorem~\ref{thm:single-parameter}.

\begin{theorem}
\label{thm:single-parameter_weak}
    Consider any single-paramter setting.
    Any \altIC, IR, and \altNB{} mechanism $M = (f, p)$ can be written as an order-oblivious decision list with exceptions almost everywhere.
    In particular, the mechanism may deviate from the decision list only when $v_i = p_i(o)$ for some agent $i$ and outcome $o$, or when $v_i = 0$ for some agent $i$.
\end{theorem}

We sketch the key differences in the proof.
Recall that in the proof of Theorem~\ref{thm:single-parameter}, we recursively construct a decision tree roughly in the following way: we start at the root with all feasible outcomes.
We pick the ``maximum'' of these outcomes $v^\cO$, where for any set of outcomes $S \subseteq \cO$, $v^S$ is defined such that $v^S_i = \max_{o \in S} p_i(o)$.
We then let $f(v^\cO)$ be the outcome associated with the root (with $\cO$ being the root's domain).
Then, for any maximal subset $D \subseteq \cO$ such that $v^D$ does not satisfy $f(v^\cO)$, we create a child of the root (with $D$ being its domain), and let the outcome associated with that child be $f(v^D)$.
We repeat this procedure recursively, which produces a decision tree.
Lemma~\ref{lem:consistency} ensures that if a valuation vector $v$ (1) satisfies the outcome $o$ associated with a node and (2) does not satisfy any outcome not in the domain of the node, then we must have $f(v) = o$.
Moreover, by construction, each valuation vector $v$ satisfies the above for at least one node of the tree.
This allows us to turn the tree into an order-oblivious decision list with exceptions.

Now with the mechanism being \altIC{} and \altNB, we can only use the weaker version of Lemma~\ref{lem:consistency}, Lemma~\ref{lem:weak_consistency}.
This requires us to properly replace ``satisfaction'' in the above argument with ``strict satisfaction''.
One concrete way to do so is the following: let $\eps > 0$ be a small parameter to be fixed later.
We extend the definition of $v^S$, and define $v^{S, \eps}$ such that for each agent $i$, $v^{S, \eps}_i = \max_{o \in S} p_i(o) + \eps$.
Then, in the construction of the tree, we replace $v^\cO$ with $v^{\cO, \eps}$, and recursively do the same for $v^D$ at every node.
Now Lemma~\ref{lem:weak_consistency} ensures that if a valuation vector $v$ (1) strictly satisfies the outcome $o$ associated with a node and (2) does not satisfy any outcome not in the domain of the node, then we must have $f(v) = o$.
The difference is that there are valuation vectors that are not covered by any node of the tree --- these are the valuation vectors that fall in the ``gaps'' between $v^D$ and $v^{D, \eps}$ at every node.
As a result, the decision list constructed from the tree does not apply to these valuation vectors.
The volume of valuation vectors that fall into the gap at every node of the tree is $n \cdot \eps$.
Since the number of nodes in the tree is finite (say $K$), the total volume of valuation vectors not covered by the decision list is at most $n \cdot K \cdot \eps$.
Now letting $\eps \to 0$, the volume of valuation vectors not covered vanishes.
In fact, the only valuation vectors left uncovered involve either ties or valuations of $0$.

\section{Restricting Possible Secondary Goals}\label{app:restricted_externalities}
In our definition of agents with secondary goals in Section \ref{sec:prelim_ext}, we allowed the external component of their preference to be any arbitrary ordering over the outcomes and payments of the other agents. In this appendix, we show that this freedom is unnecessary --- all our results continue to hold with only mild constraints on the set of possible secondary goals (namely, as long as each agent can express either altruistic or malicious preferences towards other agents, our characterization continues to hold).

We formalize this as follows. Fixing other agents' valuations to be $v_{-i}$, let $\cP_i^\pub(v_{-i})$ be the set of allowed preferences for agent $i$ over other agents' outcome-payment pairs. For any agent $j \ne i$, we say $\cP_i^\pub(v_{-i})$ {\em covers} $\mathbin{\succ_j^\pri}$, if at least one of the following two conditions hold:
\begin{itemize}
    \item There exists a preference $\mathbin{\succ_i^\pub} \in \cP_i^\pub(v_{-i})$, such that for all $(o, q) \in \cO \times \bR_+^n$, if $(o_j, q_j) \succ_j^\pri (o_j', q_j')$, then $(o_{-i}, q_{-i}) \succ_i^\pub (o_{-i}', q_{-i}')$.
    When $i$ has such an external component of preference, we say {\em $i$ cares positively about $j$}.
    \item There exists a preference $\mathbin{\succ_i^\pub} \in \cP_i^\pub(v_{-i})$, such that for all $(o, q) \in \cO \times \bR_+^n$, if $(o_j, q_j) \succ_j^\pri (o_j', q_j')$, then $(o_{-i}', q_{-i}') \succ_i^\pub (o_{-i}, q_{-i})$.
    When $i$ has such an external component of preference, we say {\em $i$ cares negatively about $j$}.
\end{itemize}

We make the following richness assumption on the possible secondary goals for each agent.

\begin{assumption}[Richness of Secondary Goals]
\label{assumption:richness}
    For any agent $i$ and $v_{-i} \in \cV_{-i}$, $\cP_i^\pub(v_{-i})$ covers all $\{\mathbin{\succ_j^\pri}\}_{j \ne i}$.
\end{assumption}

We also remark that Assumption~\ref{assumption:richness} is in a sense necessary for the concept of IC to distinguish itself from PIC as a stronger notion of robustness --- without the richness assumption, it could be the case that all agents are always indifferent to others' interests, and as a result IC would be precisely equivalent to PIC. With Assumption~\ref{assumption:richness}, all of our proofs (most notably our characterization of IC mechanisms as nonbossy in Theorem \ref{thm:semantic} continue to hold).

\section{Nonbossiness and Obvious Strategyproofness}\label{app:osp}

The notion of nonbossiness is reminiscent of another important notion of robustness against strategic behavior, namely obvious strategyproofness (OSP), proposed by \citet{li2017obviously}.
Roughly speaking, a mechanism (as an extensive-form game) is OSP, if for any agent, the worst (over other agents' actions) thing that may happen under truthful reporting is at least as good as the best thing that may happen under any deviation from reporting truthfully.
A mechanism, as a pair of allocation and payment rules, is OSP implementable, if there is a way to implement this mechanism using an extensive-form game that is OSP.
OSP implementability appears conceptually related to NB, since both notions require that it is ``hard'' for an agent to manipulate the utility of another agent.

We remark that similar as nonbossiness and OSP appear, neither of the two notions is stronger or weaker than the other, even when restricted to PIC and IR mechanisms. 

\begin{proposition}
There exists an NB, PIC and IR mechanism that is not OSP implementable, and an OSP implementable, PIC and IR mechanism that is not NB.
\end{proposition}
\begin{proof}
For the former, consider the mechanism used in the proof of Proposition~\ref{prop:payment_counterexample}, and possible ways to implement it using an extensive-form game.
Since there are only $2$ agents, each with two possible types, the mechanism should interact with each agent exactly once.
Without loss of generality, suppose the mechanism interacts with agent $1$ first (the $2$ agents are symmetric).
Consider the case where agent $1$'s valuation is $x_1 \in \cV_1$.
If agent $1$ reports $x_1$ truthfully, then the worst thing that can happen is agent $2$ reporting $y_1$, in which case agent $1$ gets outcome $o_1$ and pays $1$, leading to a utility of $\eps$.
However, if agent $1$ deviates and reports $x_2$, then the best thing that can happen is agent $2$ reporting $y_2$, in which case agent $1$ gets outcome $o_1$ and pays $0$, leading to a utility of $1 + \eps$.
This means the mechanism is not OSP implementable.

For the latter, it is known that the second-price auction is OSP implementable (using the English auction), but it is not NB, since the payment of the winner depends on the highest other bid. 
\end{proof}

Despite this non-comparability, many of the NB mechanisms we study in this paper are indeed also OSP implementable. Specifically, sequential posted-price mechanisms for a single item are OSP in their default implementation, as are PIC decision list mechanisms with empty exception lists.

\section{Single-Parameter Characterizations: Overview of Proof Techniques}
\label{app:single-parameter_overview}

Before diving into the proofs, we first present an example illustrating what these decision lists look like.

\begin{example}
    Consider an environment with $n = 3$ agents and $\cO = \{0, 1\}^3$.
    Consider the following mechanism: allocate to all $3$ agents if they all have value at least $1$; otherwise, allocate to the ``clockwise first'' agent with value at least $1$ if there exists one; otherwise, do not allocate at all.
    The full mechanism is described by the table below (the part corresponding to allocating to the ``clockwise first'' agent is highlighted in red).
    \[
        \begin{array}{ccc}
            (v_1, v_2, v_3) & (o_1, o_2, o_3) & (p_1, p_2, p_3) \\\\
            (\ge 1, \ge 1, \ge 1) & (1, 1, 1) & (1, 1, 1) \\
            \color{red} (\ge 1, \ge 1, < 1) & \color{red} (1, 0, 0) & (1, 0, 0) \\
            \color{red} (\ge 1, < 1, \ge 1) & \color{red} (0, 0, 1) & (0, 0, 1) \\
            \color{red} (< 1, \ge 1, \ge 1) & \color{red} (0, 1, 0) & (0, 1, 0) \\
            (\ge 1, < 1, < 1) & (1, 0, 0) & (1, 0, 0) \\
            (< 1, \ge 1, < 1) & (0, 1, 0) & (0, 1, 0) \\
            (< 1, < 1, \ge 1) & (0, 0, 1) & (0, 0, 1) \\
            (< 1, < 1, < 1) & (0, 0, 0) & (0, 0, 0)
        \end{array}.
    \]
    One may check that the above mechanism is PIC, IR and NB.
    The mechanism corresponds to the following decision list (which is not unique):
    \begin{itemize}
        \item $o^{(1)} = (1, 1, 1)$, $p(o^{(1)}) = (1, 1, 1)$, $E(o^{(1)}) = \emptyset$.
        \item $o^{(2)} = (1, 0, 0)$, $p(o^{(2)}) = (1, 0, 0)$, $E(o^{(2)}) = \{o^{(3)} = (0, 0, 1)\}$.
        \item $o^{(3)} = (0, 0, 1)$, $p(o^{(3)}) = (0, 0, 1)$, $E(o^{(3)}) = \{o^{(4)} = (0, 1, 0)\}$.
        \item $o^{(4)} = (0, 1, 0)$, $p(o^{(4)}) = (0, 1, 0)$, $E(o^{(4)}) = \{o^{(2)} = (1, 0, 0)\}$.
        \item $o^{(5)} = (0, 0, 0)$, $p(o^{(5)}) = (0, 0, 0)$, $E(o^{(5)}) = \emptyset$.
    \end{itemize}
    In particular, exception lists are necessary in this example: there is no way to order $o^{(2)}$, $o^{(3)}$ and $o^{(4)}$ to always choose the clockwise first agent with value at least $1$ without exception lists.
    Also note that the above exception lists are order-specific, i.e., reordering the outcomes may result in a decision list implementing a different mechanism.
    However, one can make the exception lists order-oblivious by adding $o^{(1)}$ to the exception lists of all other outcomes, and $o^{(2)}$, $o^{(3)}$ and $o^{(4)}$ to the exception list of $o^{(5)}$.
\end{example}

We now describe some of the ideas in the proofs of Theorems~\ref{thm:single-item}~and~\ref{thm:single-parameter}.
We first sketch how Theorem~\ref{thm:single-parameter} implies Theorem~\ref{thm:single-item}.

\begin{proof}[Proof Sketch of Theorem~\ref{thm:single-item}]
By Theorem 5, we know any nonbossy mechanism for the single-item setting can be written as a decision list mechanism. To show it is a posted-price mechanism, we need to show that this decision list mechanism can be written in such a way that all the exception lists are empty. 

To do this, we make use of the fact that nonbossiness imposes some additional strong constraints on the structure of the extension lists. For example, in this single-item setting we can show there cannot exist two outcomes $o$ and $o'$ (where agents $i$ and $i'$ get allocated the items) such that $o \in E(o')$ and $o' \in E(o)$. To see why, assume that $o$ comes before $o'$ in the decision list. Then if outcome $o$ occurs, $i'$ can cause a different outcome to occur by increasing their value and causing $o'$ to be feasible (all the while, $o'$ will never be selected since $o$ is feasible and $o \in E(o')$). By applying similar logic, we can construct a total ordering over all the outcomes (where $o$ dominates $o'$ if whenever $o$ is feasible, $o'$ cannot be picked) and show that there exists a posted-price mechanism which presents the items in this order. See the full proof in Appendix~\ref{sec:omitted} for details.
\end{proof}

We now discuss our proof of Theorem~\ref{thm:single-parameter}.
Our proof crucially relies on the following lemma, which shows that if two valuation vectors $v$ and $v'$ both satisfy (in the sense of Theorem~\ref{thm:single-parameter}) outcomes $o$ and $o'$, then it is not possible that under valuation $v$ the outcome is $o$ and under $v'$ the outcome is $o'$. 
\begin{lemma}
\label{lem:consistency}
    Consider any PIC, IR, and NB mechanism $M = (f, p)$ in a single-parameter environment.
    Fix two different feasible outcomes $o$ and $o'$, and let $p(o)$, $p(o')$ be the corresponding payment vectors.
    Choose any two valuation vectors $v$ and $v'$ where both $v$ and $v'$ satisfy both $o$ and $o'$ (i.e., $\min\{v_i, v'_i\} \ge \max\{p_i(o), p_i(o')\}$ for each $i \in [n]$).\footnote{
        Note that $\min$ is well defined on $\eR$.
    }
    Then,
    \[
        f(v) = o \implies f(v') \ne o'.
    \]
\end{lemma}

Intuitively, Lemma \ref{lem:consistency} follows from yet another consequence of nonbossiness and PIC; we show that if we increase the valuation from either $v$ or $v'$ to a common maximum $v''$, the outcome of the mechanism cannot change without violating one of these two properties.

Lemma \ref{lem:consistency} is useful for the following reason: instead of having to worry about the behavior of our mechanism $M = (f, p)$ for all valuation vectors $v$, we can use Lemma \ref{lem:consistency} to construct a finite number of different ``classes'' of valuation vector, where $M$ acts identically on every vector in a class. 

For example, consider the valuation vector $v^{\cO}$ defined via $v^{\cO}_i = (\max_{o \in \cO} p_i(o))_+$; in words, $v^{\cO}_i$ is the minimal valuation vector which satisfies every outcome in $\cO$. By Lemma \ref{lem:consistency}, if $v$ is a valuation vector which satisfies the outcome $f(v^{\cO})$, then $v$ must also lead to the outcome $f(v^{\cO})$. We can therefore group the set of all valuations $v$ which satisfy $f(v^{\cO})$ into a single ``class''.

In the proof of Theorem \ref{thm:single-item} we repeat this logic for other subsets of outcomes to first construct a structure we call a ``decision tree''. In a decision tree, every node is labeled by a subset $S \subseteq \cO$ of outcomes and corresponds to a class of valuations ``equivalent to'' (in the sense that we can apply Lemma \ref{lem:consistency} to argue they have the same outcome) the minimal valuation $v^{S}$ that satisfies every outcome in $S$. The root of this tree is labeled by $\cO$ and contains the class of valuations described above. The children of a node labeled $S$ contain the maximal subsets $S'$ of $S$ where $v^{S'}$ does not belong to the class of valuations of $v^S$. 

By constructing this tree in this way, we can prove that each valuation $v$ is associated to some node, and in fact it is associated to the node with the maximal label $S$ such that $v$ does not satisfy any outcome outside of $S$. This gives us a ``decision-list-esque'' mechanism, where for each $v$ we iterate through the nodes of our tree until we find a node $S$ where $v$ both satisfies $f(v^S)$ and does not satisfy any outcome $o \not\in S$. This process may consider the same outcomes more than once, but with a little bit more care we can convert this to a single decision list, and even an order-oblivious one (see the full proof in Appendix~\ref{sec:omitted} for details).

\section{Omitted Proofs}\label{sec:omitted}

\subsection{Proof of Theorem~\ref{thm:semantic}}

\begin{proof}[Proof of Theorem~\ref{thm:semantic}]
    First, fix an environment $(\cV, \cO)$ and a mechanism $M = (f, p)$.
    Assuming $M$ is PIC and NB, we show it is IC.
    To this end, consider any agent $i$, valuation vector $v \in \cV$, external component of $i$'s preference $\mathbin{\succ_i^\pub} \in \cP_i^\pub(v_{-i})$, and deviation $v_i' \in \cV_i$.
    Also let $\succ_i^\pri$ be the internal component of $i$'s preference induced by $v_i$.
    Since $M$ is PIC, we have
    \[
        (f_i(v_i, v_{-i}), p_i(v_i, v_{-i})) \succeq_i^\pri (f_i(v_i', v_{-i}), p_i(v_i', v_{-i})).
    \]
    So, in order to show
    \[
        (f(v_i, v_{-i}), p(v_i, v_{-i})) \succeq_i (f(v_i', v_{-i}), p(v_i', v_{-i})),
    \]
    we only need to argue that whenever
    \[
        (f_i(v_i, v_{-i}), p_i(v_i, v_{-i})) = (f_i(v_i', v_{-i}), p_i(v_i', v_{-i})),
    \]
    we have
    \[
        (f_{-i}(v_i, v_{-i}), p_{-i}(v_i, v_{-i})) \succeq_i^\pub (f_{-i}(v_i', v_{-i}), p_{-i}(v_i', v_{-i})).
    \]
    This follows directly from the fact that $M$ is NB, since when
    \[
        (f_i(v_i, v_{-i}), p_i(v_i, v_{-i})) = (f_i(v_i', v_{-i}), p_i(v_i', v_{-i})),
    \]
    NB requires that
    \[
        (f(v_i, v_{-i}), p(v_i, v_{-i})) = (f(v_i', v_{-i}), p(v_i', v_{-i})),
    \]
    and therefore
    \[
        (f_{-i}(v_i, v_{-i}), p_{-i}(v_i, v_{-i})) =_i^\pub (f_{-i}(v_i', v_{-i}), p_{-i}(v_i', v_{-i})).
    \]

    Now consider the other direction.
    Assuming $M$ is IC, we show it is PIC and NB.
    First observe that any IC mechanism $M$ is PIC.
    Below we show $M$ is also NB.
    Fix any agent $i$, valuation vector $v$, and deviation $v_i'$.
    Suppose towards a contradiction that there exists $j \ne i$, such that
    \[
        (f_j(v_i, v_{-i}), p_j(v_i, v_{-i})) \ne (f_j(v_i', v_{-i}), p_j(v_i', v_{-i})).
    \]
    Moreover, without loss of generality,
    \[
        (f_j(v_i', v_{-i}), p_j(v_i', v_{-i})) \succ_j^\pri (f_j(v_i, v_{-i}), p_j(v_i, v_{-i})).
    \]
    Now, let us choose an external preference component $\mathbin{\succ_i^\pub}$ for agent $i$ such that for all $(o, q) \in \cO \times \bR_+^n$, if $(o_j, q_j) \succ_j^\pri (o_j', q_j')$, then $(o_{-i}, q_{-i}) \succ_i^\pub (o_{-i}', q_{-i}')$. Note that even if we are in the restricted model of secondary goals of Appendix \ref{app:restricted_externalities},  Assumption~\ref{assumption:richness} guarantees such an ordering exists in $\cP_i^\pub(v_{-i})$. When $i$ has this external preference component, we have
    \[
        (f(v_i', v_{-i}), p(v_i', v_{-i})) \succ_i (f(v_i, v_{-i}), p(v_i, v_{-i})).
    \]
    However, RwSG requires that
    \[
        (f(v_i, v_{-i}), p(v_i, v_{-i})) \succeq_i (f(v_i', v_{-i}), p(v_i', v_{-i})),
    \]
    a contradiction.
    This concludes the proof.
\end{proof}

\subsection{Proof of Proposition~\ref{prop:payment_applicability}}

\begin{proof}[Proof of Proposition~\ref{prop:payment_applicability}]
    We argue these cases one by one.
    In each of these settings, fixing $i$, $v_i$, $v_i'$ and $o_i$, we explicitly construct a common upper bound of $v_i$ and $v_i'$ with respect to $o_i$, as described in Definition~\ref{def:upper_semilattice}.
    \begin{itemize}
        \item
        {\bf Common outcome, complete domain.}
        Let $v_i''$ be such that $v_i''(o_i) = \max\{v_i(o_i), v_i'(o_i)\}$ and $v_i''(o_i') = 0$ for $o_i' \in \cO_i \setminus \{o_i\}$.
        \item 
        {\bf Single-parameter agents.} 
        If $o_i = 0$, then let $v_i''(1) = 0$; otherwise, let $v_i''(1) = \max \{v_i(1), v_i'(1)\}$.
        \item
        {\bf Combinatorial auctions, single-minded agents.}
        Let $v_i$ be described by $S_i$ and $x_i$, and $v_i'$ by $S_i'$ and $x_i'$.
        Let $v_i''$ be induced by $S_i''$ and $x_i''$, where $S_i'' = o_i$, and $x_i'' = \max\{x_i, x_i'\}$.
        Consider $v_i$, and observe that for any $o_i' \subseteq M$, if $o_i \not\subseteq o_i'$, then
        \[
            v_i''(o_i) - v_i''(o_i') = v_i''(o_i) \ge v_i(o_i) \ge v_i(o_i) - v_i(o_i').
        \]
        If $o_i \subseteq o_i'$, then because $v_i$ is monotone,
        \[
            v_i''(o_i) - v_i''(o_i') = 0 \ge v_i(o_i) - v_i(o_i').
        \]
        Similarly, for $v_i'$, we always have
        \[
            v_i''(o_i) - v_i''(o_i') \ge v_i'(o_i) - v_i'(o_i').
        \]
        \item
        {\bf Combinatorial auctions, valuations ``between'' additive and XOS.}
        Recall that an XOS valuation is the pointwise maximum of a number of additive valuations, each of which is a clause.
        Let $c_i$ be the clause of $v_i$ such that $v_i(o_i) = c_i(o_i)$, and $c_i'$ be the clause of $v_i'$ such that $v_i'(o_i) = c_i'(o_i)$.
        Note that $c_i$ (resp.\ $c_i'$) also satisfies for any set of items $T \subseteq M$, $c_i(T) \le v_i(T)$ (resp.\ $c_i'(T) \le v_i'(T)$).
        Let $v_i''$ be an additive valuation such that for each item $j$,
        \[
            v_i''(\{j\}) = \begin{cases}
            \max\{c_i(\{j\}), c_i'(\{j\})\}, & \text{if } j \in o_i \\
            0, & \text{otherwise}
            \end{cases}.
        \]
        Clearly $v_i'' \in \cV_i$, because it is additive.
        For any $o_i' \subseteq M$,
        \begin{align*}
            v_i''(o_i) - v_i''(o_i') & = \sum_{j \in o_i \setminus o_i'} v_i''(\{j\}) - \sum_{j \in o_i' \setminus o_i} v_i''(\{j\}) \\
            & = \sum_{j \in o_i \setminus o_i'} v_i''(\{j\}) \ge \sum_{j \in o_i \setminus o_i'} c_i(\{j\}) \tag{construction of $v_i''$} \\
            & \ge \sum_{j \in o_i \setminus o_i'} c_i(\{j\}) - \sum_{j \in o_i' \setminus o_i} c_i(\{j\}) \\
            & = c_i(o_i) - c_i(o_i') = v_i(o_i) - c_i(o_i') \tag{property of $c_i$} \\
            & \ge v_i(o_i) - v_i(o_i'). \tag{property of $c_i$}
        \end{align*}
        Similarly,
        \[
            v_i''(o_i) - v_i''(o_i') \ge v_i'(o_i) - v_i'(o_i').
        \]
        \item
        {\bf Combinatorial auctions, valuations ``beyond'' subadditive.}
        Let $C = \max\{v_i(M), v_i'(M)\}$, and $v_i''$ be such that $v_i''(\emptyset) = 0$, $v_i''(T) = 2C$ if $o_i \subseteq T$, and $v_i''(T) = C$ otherwise.
        It is easy to check $v_i''$ is subadditive, and so $v_i'' \in \cV_i$.
        For any $o_i'$, if $o_i \subseteq o_i'$, because $v_i$ is monotone,
        \[
            v_i''(o_i) - v_i''(o_i') = 0 \ge v_i(o_i) - v_i(o_i').
        \]
        If $o_i \not\subseteq o_i'$, then
        \[
            v_i''(o_i) - v_i''(o_i') = C \ge v_i(M) \ge v_i(o_i) - v_i(o_i').
        \]
        Similarly,
        \[
            v_i''(o_i) - v_i''(o_i') \ge v_i'(o_i) - v_i'(o_i').
        \]
        \item
        {\bf Metric space.}
        Let $v_i$ be induced by $x_i$ and $v_i'$ by $x_i'$.
        Let $v_i''$ be induced by $o_i \in \cO \subseteq X_i$.
        Then for any $o_i' \in \cO \subseteq X_i$,
        \[
            v_i''(o_i) - v_i''(o_i') = d_i(o_i, o_i') \ge d_i(x_i, o_i') - d_i(x_i, o_i) = v_i(o_i) - v_i(o_i').
        \]
        And similarly
        \[
            v_i''(o_i) - v_i''(o_i') \ge v_i'(o_i) - v_i'(o_i'). \qedhere
        \]
    \end{itemize}
\end{proof}

\subsection{Proof of Theorem~\ref{thm:common}}

\begin{proof}[Proof of Theorem~\ref{thm:common}]
    Fix a PIC and NB mechanism $M = (f, p)$.
    By Proposition~\ref{prop:payment_applicability}, the payment characterization (Theorem~\ref{thm:payment}) applies in this setting.
    To this end, for each $o_c \in \cO_1 = \dots = \cO_n$, let $p_i(o_c)$ be agent $i$'s payment when the common outcome is $o_c$.
    For each valuation vector $v$, consider the following transformed valuation vector $\hat{v}$, where for each agent $i$ and outcome $o_c$,
    \[
        \hat{v}_i(o_c) = v_i(o_c) - p_i(o_c).
    \]
    Observe that for any $i$ and $v_i \in \cV_i$, for any $o_c, o_c' \in \cO_1 = \dots = \cO_n$,
    \[
        (o_c, p_i(o_c)) \succ_i^\pri (o_c', p_i(o_c')) \iff \hat{v}_i(o_c) > \hat{v}_i(o_c') \vee (\hat{v}_i(o_c) = \hat{v}_i(o_c') \wedge o_c \vartriangleright_i o_c'),
    \]
    where $\vartriangleright_i$ is the tiebreaking component of $v_i$.
    Now consider the behavior of $\hat{f}$, the transformed version of $f$, over all possible transformed valuations $\hat\cV = \{\hat{v} \mid v \in \cV\}$, where we define $\hat{f}(\hat{v}) = f(v)$.
    Then, the above observation gives a way to derive $\mathbin{\succ_i^\pri}$ from $\hat{v}$ directly, and so we can define PIC for $\hat{f}$ based on $\succ_i^\pri$.
    Note that $\hat{f}$ corresponds bijectively to $f$, and $(\hat{f}, p^\mathrm{zero})$ is PIC iff $(f, p)$ is PIC, where $p^\mathrm{zero}$ assigns $0$ payment to all agents in all cases.
    In other words, one can view $\hat{f}$ as a PIC mechanism over $\hat\cV$ without payments.
    Suppose $|\cO_1| = |\cO_2| = \dots = |\cO_n| = k$.
    Observe that such transformed valuations $\hat\cV$ can induce all possible strict total orders over outcomes: for example, in order for agent $i$'s internal component of preference to be
    \[
        o_c^{(1)} \succ_i^\pri o_c^{(2)} \succ_i^\pri \dots \succ_i^\pri o_c^{(k)},
    \]
    we only need
    \[
        \hat{v}_i(o_c^{(1)}) > \hat{v}_i(o_c^{(2)}) > \dots > \hat{v}_i(o_c^{(k)}).
    \]
    Such a transformed valuation can be obtained from, for example, $v_i \in \cV_i$ such that
    \[
        v_i = (k - i) + p_i(o_c^{(i)}).
    \]
    Now since $\hat{f}$ is PIC, it must induce a social choice rule over all strict total orders over the outcomes (i.e., it cannot assign different outcomes to two transformed valuation vectors inducing the same orders for all agents simultaneously), and moreover, this social choice rule cannot be manipulable.
    By the Gibbard-Satterthwaite Theorem \citep{gibbard1973manipulation,satterthwaite1975strategy}, this social choice rule must either be dictatorship or use only two outcomes.
    The same characterization immediately applies to $\hat{f}$, and by the payment characterization extends to $(f, p)$, which concludes the proof.
\end{proof}

\subsection{Proof of Theorem~\ref{thm:single-item}}

\begin{proof}[Proof of Theorem~\ref{thm:single-item}]
    \sloppy{
    For any PIC, IR and NB mechanism $M = (f, p)$, let $o^{(1)}, \dots, o^{(n + 1)}$ be a decision list implementing $M$ (with payment vectors $p(o^{(1)}), \dots, p(o^{(n + 1)})$ and exception lists $E(o^{(1)}), \dots, E(o^{(n + 1)})$).
    Such a decision list exists by Theorem~\ref{thm:single-parameter} (we do not even require this list to be order-oblivious).
    Without loss of generality, assume each agent can sometimes get the item --- otherwise we can remove the agents who never get the item, since they can never affect the outcome and payments because of NB.
    Note that the empty outcome $(0, \dots, 0)$ can never appear in an exception list of another outcome $o^{(j)}$, because in that case $o^{(j)}$ can never happen (since the payment vector corresponding to the empty outcome must be $(0, \dots, 0)$, and the empty outcome is always satisfied).}
    Moreover, if the empty outcome is ordered before another outcome $o^{(j)}$, $o^{(j)}$ must be in the exception list of the empty outcome, because, again, otherwise $o^{(j)}$ can never happen.
    Given these facts, we can without loss of generality move the empty outcome to the end of the list, and assume $o^{(n + 1)} = (0, \dots, 0)$.
    Moreover, we renumber the agents so that $o^{(i)}_i = 1$ for all $i \in [n]$.
    That is, in the $i$-th outcome, agent $i$ receives the item.
    For brevity, let $q_i = p_i(o^{(i)})$ for each $i \in [n]$.

    Now observe that for any two agents $i$ and $i'$, it cannot be the case that $o^{(i)} \in E(o^{(i')})$ and $o^{(i')} \in E(o^{(i)})$ simultaneously.
    This is because if that happens, then for the valuation vector $v$ where $v_i = q_i$, $v_{i'} = q_{i'}$, and $v_{i''} = 0$ for $i'' \notin \{i, i'\}$, we must have $f(v) \notin \{o^{(i)}, o^{(i')}\}$.
    However, for the valuation vector $v'$ where $v'_i = q_i$ and $v'_{i''} = 0$ for $i'' \ne i$, we must have $f(v') = o^{(i)}$.\footnote{One subtlety is it is possible that there exists some $i''$ where $q_{i''} = 0$. This does not affect the argument much, because such an $i''$ cannot appear in the exception list associated with any other nonempty outcome, and $i''$ must be agent $n$. Then one can assume without loss of generality that $i' \ne n$, and the rest of the argument still works. Similarly one can fix the argument for the nonexistence of $3$-cycles below.}
    In other words, $i'$ can affect $i$'s outcome and payment without affecting $i'$'s own outcome or payment, which violates NB.
    
    We define the following ``domination'' binary relation over agents: we say $i$ dominates $i'$ if $o^{(i)} \in E(o^{(i')})$, or $i < i'$ and $o^{(i')} \notin E(o^{(i)})$. Intuitively, if $i$ dominates $i'$, then if $o^{(i)}$ is feasible (i.e., $v_i \geq q_i$) then $o^{(i')}$ cannot be chosen. Observe that any two different agents $i$ and $i'$ are comparable under the domination relation, i.e., one of the two must dominate the other, and it is never the case that the two agents dominate each other simultaneously (because it cannot be the case that $o^{(i)} \in E(o^{(i')})$ and $o^{(i')} \in E(o^{(i)})$ simultaneously, as argued above).

    Below we argue by contradiction that the domination relation actually is a total order over agents.
    In particular, it does not have $3$-cycles.
    Suppose otherwise, i.e., there exist $i$, $i'$ and $i''$, such that $i$ dominates $i'$, $i'$ dominates $i''$, and $i''$ dominates $i$.
    Then for the valuation vector $v$ where $v_i = q_i$, $v_{i'} = q_{i'}$, $v_{i''} = q_{i''}$ and $v_{i'''} = 0$ for all $i''' \notin \{i, i', i''\}$,
    we must have $f(v) \notin \{o^{(i)}, o^{(i')}, o^{(i'')}\}$.
    However, for $v'$ where $v'_i = q_i$, $v'_{i'} = q_{i'}$ and $v'_{i'''} = 0$ for all $i''' \notin \{i, i'\}$, we must have $f(v') = o^{(i)}$.
    This means $i''$ can affect $i$'s outcome and payment without affecting $i''$'s own outcome or payment, violating NB.

    Now we argue that the total order given by the domination relation is precisely one possible order that can be used in a sequential posted-price mechanism to implement $M$.
    In particular, the following sequential posted-price mechanism is equivalent to the decision list implementing $M$: repeatedly visit the agent who dominates all other agents among unvisited ones.
    Upon visiting each agent $i$, offer the price $q_i$.
    If $v_i \ge q_i$, then let $f(v) = o^{(i)}$ and $p(v) = p(o^{(i)})$.
    Otherwise, continue the process.
    Terminate the process and let $(f(v), p(v)) = ((0, \dots, 0), (0, \dots, 0))$ if there is no unvisited agent left.
    One can check this sequential posted-price mechanism does precisely what the decision list does.
\end{proof}

\subsection{Proof of Lemma~\ref{lem:consistency}}

\begin{proof}[Proof of Lemma~\ref{lem:consistency}]
    Suppose to the contrary that both $f(v) = o$ and $f(v') = o'$.
    Choose $v''$ such that $v''_i = \min\{v_i, v'_i\}$ for each $i \in [n]$.
    The plan is to argue that $f(v'') = f(v)$ and $f(v'') = f(v')$ simultaneously, a contradiction.
    We only need to show $f(v'') = f(v) = o$, for which it suffices to inductively argue that for each $i \in [n]$,
    \[
        f(v) = f(v_1, \dots, v_{i - 1}, v''_i, v''_{i + 1}, \dots, v''_n) = f(v_1, \dots, v_{i - 1}, v_i, v''_{i + 1}, \dots, v''_n).
    \]
    Fix some $i \in [n]$.
    By the induction hypothesis,
    \[
        f(v_1, \dots, v_{i - 1}, v_i, v''_{i + 1}, \dots, v''_n) = f(v) = o.
    \]
    Recall that $v_i \ge v''_i$ by the choice of $v''$.
    Consider the following $2$ cases:
    \begin{itemize}
        \item $o_i = 0$.
        In this case, because $M$ is PIC, $f$ must be monotone (by Myerson's characterization of single-parameter PIC mechanisms \citep{myerson1981optimal}), so decreasing agent $i$'s value from $v_i$ to $v''_i$ cannot change the personal outcome (or payment) of agent $i$, i.e.,
        \[
            f_i(v_1, \dots, v_{i - 1}, v''_i, v''_{i + 1}, \dots, v''_n) = f_i(v_1, \dots, v_{i - 1}, v_i, v''_{i + 1}, \dots, v''_n).
        \]
        Now since $M$ is NB, the above implies
        \[
            f(v_1, \dots, v_{i - 1}, v''_i, v''_{i + 1}, \dots, v''_n) = f(v_1, \dots, v_{i - 1}, v_i, v''_{i + 1}, \dots, v''_n).
        \]
        \item $o_i = 1$.
        In this case, because $M$ is PIC, fixing other agents' valuations, the menu agent $i$ faces is a take-it-or-leave-it offer with price $p_i(o) \le v''_i$ (again by Myerson's characterization).
        So, decreasing $i$'s value from $v_i$ to $v''_i$ results in the same outcome and payment, because both $v_i$ and $v''_i$ are no smaller than $p_i(o)$.
        Again, since $M$ is NB, this implies
        \[
            f(v_1, \dots, v_{i - 1}, v''_i, v''_{i + 1}, \dots, v''_n) = f(v_1, \dots, v_{i - 1}, v_i, v''_{i + 1}, \dots, v''_n).
        \]
    \end{itemize}
    This finishes the induction step, and shows that $f(v) = f(v'')$.
    But then similarly one can show $f(v') = f(v'')$, which leads to a contradiction because $f(v) = o \ne o' = f(v')$.
    This concludes the proof.
\end{proof}

\subsection{Proof of Theorem~\ref{thm:single-parameter}}

\begin{proof}[Proof of Theorem~\ref{thm:single-parameter}]
    Given a PIC, IR, and NB mechanism $M = (f, p)$, we explicitly construct a decision list implementing this mechanism.
    We first construct a ``decision tree'' (which is quite different from a normal decision tree), which we will later turn into a decision list.
    The decision tree we construct is of the following form: each node is associated with a set of outcomes, which we call the domain of the node, and a single outcome in the domain.
    In particular, the domain of the root is $\cO$.
    Each node may have any number of children, with the constraint that the domain of any child of a node is a strict subset of the domain of its parent (so the decision tree must be finite).
    The way we choose an outcome using such a decision tree is the following: starting from the root, at every node, we check whether the valuation vector $v$ satisfies the outcome associated with the node. 
    If yes, then that outcome is the one we choose.
    Otherwise, we move on to an arbitrary child of the current node satisfying the following condition: $v$ does not satisfy any outcome that is not in the domain of that child (we will see from the construction that such a child always exists), and repeat the above procedure.
    The intuition behind the procedure will become clear momentarily.

    Now we construct the tree.
    For any set of outcomes $S \subseteq \cO$, let $v^S$ be such that for each $i \in [n]$,
    \[
        v^S_i = (\max\{p_i(o) \mid o \in S\})_+.
    \]
    We start from the root, whose domain is $\cO$, and let $f(v^\cO)$ be the outcome associated with the root.
    This means the tree should choose $f(v^\cO)$ whenever it is satisfied.
    This is in fact consistent with $f$, because for any other outcome $o \in \cO$, $v^\cO$ satisfies both $f(v^\cO)$ and $o$.
    By Lemma~\ref{lem:consistency}, whenever $f(v^\cO)$ is satisfied by a valuation vector $v$, $f(v) \ne o$.
    This holds for any $o \in \cO \setminus \{f(v^\cO)\}$, so it must be the case that $f(v) = f(v^\cO)$ whenever $v$ satisfies $f(v^\cO)$.

    Now we describe the way we construct the children of a node (e.g., the root), which, applied recursively, gives the entire tree.
    Let $D$ be the domain of the parent node, and $o$ be the outcome associated with the parent node.
    Each child corresponds to a maximal subset $D'$ of $D$ such that $v^{D'}$ does not satisfy $o$.
    Note that there may be many such maximal subsets, and we construct a child for each of these maximal subsets.
    For the child whose domain is $D'$, again, we let the outcome associated with that child be $f(v^{D'})$.
    Observe that when the tree chooses $f(v^{D'})$ at this node, this choice must also be consistent with $f$.
    This is because if we reach this node, then it must be the case that the valuation vector $v$ can only satisfy outcomes in $D'$.
    And again by Lemma~\ref{lem:consistency}, if $v$ satisfies $f(v^{D'})$, then it must be the case that $f(v) = f(v^{D'})$.

    Now consider the correctness of the decision tree constructed, i.e., whether it implements $f$.
    Above we have argued that whenever the decision tree chooses an outcome, that outcome must be consistent with the choice by $f$.
    So we only need to show the decision tree in fact always outputs an outcome.
    Consider any valuation vector $v$.
    First observe that $v$ satisfies at least $1$ outcome, i.e., the empty outcome $(0, \dots, 0)$.
    Below we argue that at any node with domain $D$ and associated outcome $o$, there is always a child that we can go to, if $v$ does not satisfy $o$.
    In fact, let $S \ne \emptyset$ be the set of outcomes that $v$ satisfies.
    Because of the way we go down the tree, we must have $S \subseteq D$.
    Also, we have $v_i \ge v^S_i$ for each $i \in [n]$, simply by the definition of $S$.
    If $v$ does not satisfy $o$, then clearly $v^S$ does not satisfy $o$ either.
    So $S$ must be a subset of the domain of at least $1$ child of the current node, because we choose the domains of the children in a maximal way.
    That child is a node we can go to.
    Now because the tree is finite, we cannot keep going down forever, so we must be able to choose an outcome somewhere in the tree.
    Together with all the arguments above, this shows that the tree we construct in fact implements $f$ (and $M$ because $p(v) = p(f(v))$).

    Now we turn the decision tree into a decision list.
    First observe that when evaluating the decision tree, the order in which we visit the nodes does not matter (although the specific way of evaluating the tree discussed above is instrumental in proving the correctness of the construction).
    More specifically, let $\{(D^{(j)}, o^{(j)})\}_{j \in [K]}$ be the collection of domain-outcome pairs of all nodes of the decision tree, where $K$ is the number of nodes.
    To evaluate the decision tree, it suffices to check $(D^{(j)}, o^{(j)})$ one by one in any order, and choose the first $o^{(j)}$ where $v$ satisfies $o^{(j)}$ and does not satisfy any outcome not in $D^{(j)}$.
    In fact, since the decision tree implements $f$, exactly one outcome from these pairs will satisfy the condition of being chosen.
    Now we merge pairs with the same outcome: for any outcome $o \in \cO$, let
    \[
        D(o) = \bigcup_{j: o^{(j)} = o} D^{(j)}.
    \]
    Consider the new collection of pairs $\{(D(o), o)\}_{o \in \cO}$, and observe that it is equivalent to the old list: when the old list chooses some outcome, the new list chooses the same outcome.
    One subtlety is that the new list may choose multiple outcomes simultaneously, which would make it ill-defined.
    However, this cannot happen because whenever the new list chooses an outcome, it must be the outcome assigned by $f$.
    To see why this is the case, consider any valuation vector $v$.
    Let $(D(o), o)$ be a pair in the new list such that $v$ satisfies $o$ and does not satisfy any outcome not in $D(o)$.
    By Lemma~\ref{lem:consistency}, because $v$ satisfies $o$, for any $o' \in D(o) \setminus \{o\}$ that is also satisfied by $v$, $f(v) \ne o'$.
    This is because by the construction of $D(o)$, there exists some $v'$ which satisfies both $o$ and $o'$ such that $f(v') = o$ (consider any node in the decision tree whose domain contains $o'$ and whose associated outcome is $o$).
    On the other hand, since $M$ is IR, $f(v)$ must be satisfied by $v$, so the only option left is $f(v) = o$.
    Now let the exception list of each outcome $o$ be $E(o) = \cO \setminus D(o)$.
    This gives an order-oblivious decision list as stated in Theorem~\ref{thm:single-parameter}.
\end{proof}

\subsection{Proof of Lemma~\ref{lem:order-oblivious_ic}}

\begin{proof}[Proof of Lemma~\ref{lem:order-oblivious_ic}]
    First note that if an order-oblivious decision list with exceptions induces a well-defined mechanism, it must be IR and NB.
    The IR part is straightforward: in order for an outcome to be chosen, it must be satisfied by the valuation vector, which means all agents receiving an item prefer paying their respective prices for an item to receiving nothing.
    Below we argue it is also NB.
    Consider any agent $i$ and fix all other agents' valuations $v_{-i}$.
    Imagine the process where agent $i$'s valuation rises from $0_-$ to $\infty$.
    Since the decision list is order-oblivious, whenever the outcome chosen by the decision list changes, it must be the case that the new outcome becomes satisfied, while some outcome in the exception list of the previously chosen outcome also becomes satisfied.\footnote{
        This is to be contrasted with the following case that can only happen when the decision list is order-specific: the previously chosen outcome encounters an exception, and the outcome chosen becomes something that comes later in the list.
        In the case of an order-oblivious decision list, this would make the mechanism ill-defined.
    }
    Then by definition, the new price thta agent $i$ pays must be strictly larger than the previous price (this also covers the case where $i$ receives no item in the previously chosen outcome).
    In other words, whenever the outcome vector changes because of $i$'s change of behavior, $i$'s payment must also change, which implies NB.

    Now we show that the induced mechanism is PIC if and only if the condition in the claim is satisfied.
    First we assume the order-oblivious decision list with exceptions induces a PIC mechanism, and argue it satisfies the condition.
    Suppose towards a contradiction that for some agent $i$ and vector $u \in \eR^n$, there exist distinct outcomes $o^{(1)}, o^{(2)} \in O^u$ such that for each $j \in \{1, 2\}$,
    \[
        o_i^{(j)} = 1 \quad \text{and} \quad \{o \in E(o^{(j)}) \cap O^u \mid p_i(o) \le p_i(o^{(j)})\} = \emptyset.
    \]
    Without loss of generality, suppose $p_i(o^{(1)}) \ge p_i(o^{(2)})$.
    We consider the behavior of the mechanism at $(v_i, v_{-i})$ and $(v_i', v_{-i})$, where:
    \begin{itemize}
        \item $v_i = (p_i(o^{(1)}))_+$,
        \item $v_i' = (p_i(o^{(2)}))_+$, and
        \item $v_{-i}$ be such that $v_{i'} = (\max_{o \in O^u} p_{i'}(o))_+$ for each $i' \ne i$.
    \end{itemize}
    Observe that for any $x \le u_i$ cannot satisfy any outcome in $\cO \setminus O^u$, because $(x, v_{-i})$ must be coordinate-wise smaller than or equal to $u$.
    Also, by definition $v_i' \le v_i \le u_i$, so neither $(v_i, v_{-i})$ nor $(v_i', v_{-i})$ can satisfy any outcome in $\cO \setminus O^u$.
    Consider the behavior of the mechanism at $(v_i, v_{-i})$.
    Observe that any $o$ satisfied by $(v_i, v_{-i})$ must satisfy $o \in O^u$ and $p_i(o) \le p_i(o^{(1)})$.
    Then by the choice of $o^{(1)}$, such an $o$ must not be a member of $E(o)$.
    Also, $(v_i, v_{-i})$ in fact satisfies $o^{(1)}$, which means the desicion list must choose $o^{(1)}$ as the outcome at $(v_i, v_{-i})$.
    In particular, agent $i$ receives an item and pays $p_i(o^{(1)})$ at $(v_i, v_{-i})$.
    By similar logic, the decision list must choose $o^{(2)}$ as the outcome at $(v_i', v_{-i})$, and agent $i$ receives an item and pays $p_i(o^{(2)})$.
    Observe that $p_i(o^{(1)}) \ne p_i(o^{(2)})$, because otherwise both $o^{(1)}$ and $o^{(2)}$ should be chosen at $(v_i, v_{-i}) = (v_i', v_{-i})$, making the mechanism ill-defined.
    But then agent $i$ would deviate from $(v_i, v_{-i})$ to $(v_i', v_{-i})$ because $i$ pays strictly less when reporting $v_i'$, violating PIC.

    Now consider the other direction, i.e., given the condition we show the mechanism is PIC.
    Fix any $i$, $v_i, v_i' \in \eR$ and $v_{-i} \in \eR^{-i}$.
    We only need to argue that $i$ would not deviate from $(v_i, v_{-i})$ to $(v_i', v_{-i})$.
    Let $u \in \eR^n$ be such that $u_{-i} = v_{-i}$ and $u_i = \max\{v_i, v_i'\}$.
    Consider the following cases:
    \begin{itemize}
        \item For each $o \in O^u$, either $o_i = 0$ or $\{o' \in E(o) \cap O^u \mid p_i(o') \le p_i(o)\} \ne \emptyset$.
        In particular, this means for all $x \le u_i$, the outcome $o$ chosen at $(x, v_{-i})$ must satisfy $o_i = 0$, which also applies to $v_i$ and $v_i'$.
        So $i$ would not deviate from $v_i$ to $v_i'$.
        \item There is precisely one $o^* \in O^u$ such that $o_i^* = 1$ and $\{o \in E(o^*) \cap O^u \mid p_i(o) \le p_i(o^*)\} \ne \emptyset$.
        Then, it is impossible for $i$ to pay different prices for an item at $v_i$ and $v_i'$, because if $i$ receives an item, the outcome chosen must be $o^*$, which uniquely determines the price.
        The only case left for $i$ to be motivated to deviate from $v_i$ to $v_i'$ is when $i$ receives no item at $v_i$ and pays less than $v_i$ for an item at $v_i'$.
        In order for this to happen, the outcome chosen at $(v_i', v_{-i})$ must be $o^*$, and we must also have $p_i(o^*) < \min\{v_i, v_i'\}$.
        This means both $(v_i, v_{-i})$ and $(v_i', v_{-i})$ satisfy $o^*$, so it must be the case that $(v_i', v_{-i})$ does not satisfy any outcome in $O^u \cap E(o^*)$ while $(v_i, v_{-i})$ does satisfy some.
        As a result, it must be the case that $v_i' < v_i$.
        However, this would make the mechanism ill-defined at $(v_i, v_{-i})$.
        This is because $o^*$ is chosen at $(v_i', v_{-i})$, so every $o' \in O^u$ where $o'_i = 0$ must have in its exception list some outcome $o'' \in E(o') \cap O^u$ satisfied by $(v_i', v_{-i})$.
        The latter part of the statement is also true at $(v_i, v_{-i})$ because $v_i > v_i'$.
        As a result, no outcome $o' \in O^u$ where $o'_i = 0$ can be chosen at $(v_i, v_{-i})$, which contradicts the fact that $i$ receives no item at $(v_i, v_{-i})$.
        In other words, $i$ would in no case deviate from $(v_i, v_{-i})$ to $(v_i', v_{-i})$.\qedhere
    \end{itemize}
\end{proof}


\begin{thebibliography}{72}
\providecommand{\natexlab}[1]{#1}
\providecommand{\url}[1]{\texttt{#1}}
\expandafter\ifx\csname urlstyle\endcsname\relax
  \providecommand{\doi}[1]{doi: #1}\else
  \providecommand{\doi}{doi: \begingroup \urlstyle{rm}\Url}\fi

\bibitem[Akbarpour and Li(2020)]{akbarpour2020credible}
Mohammad Akbarpour and Shengwu Li.
\newblock Credible auctions: A trilemma.
\newblock \emph{Econometrica}, 88\penalty0 (2):\penalty0 425--467, 2020.

\bibitem[Alva(2017)]{alva2017manipulation}
Samson Alva.
\newblock When is manipulation all about the ones and twos.
\newblock \emph{Unpublished}, 2017.

\bibitem[Babaioff et~al.(2014)Babaioff, Immorlica, Lucier, and
  Weinberg]{babaioff2014simple}
Moshe Babaioff, Nicole Immorlica, Brendan Lucier, and S~Matthew Weinberg.
\newblock A simple and approximately optimal mechanism for an additive buyer.
\newblock In \emph{2014 IEEE 55th Annual Symposium on Foundations of Computer
  Science}, pages 21--30. IEEE, 2014.

\bibitem[Bade(2020)]{bade2020random}
Sophie Bade.
\newblock Random serial dictatorship: the one and only.
\newblock \emph{Mathematics of Operations Research}, 45\penalty0 (1):\penalty0
  353--368, 2020.

\bibitem[Balkanski et~al.(2022)Balkanski, Garimidi, Gkatzelis, Schoepflin, and
  Tan]{balkanski2022deterministic}
Eric Balkanski, Pranav Garimidi, Vasilis Gkatzelis, Daniel Schoepflin, and
  Xizhi Tan.
\newblock Deterministic budget-feasible clock auctions.
\newblock In \emph{Proceedings of the 2022 Annual ACM-SIAM Symposium on
  Discrete Algorithms (SODA)}, pages 2940--2963. SIAM, 2022.

\bibitem[Barber{\`a} and Jackson(1995)]{barbera1995strategy}
Salvador Barber{\`a} and Matthew~O Jackson.
\newblock Strategy-proof exchange.
\newblock \emph{Econometrica: Journal of the Econometric Society}, pages
  51--87, 1995.

\bibitem[Bartling and Netzer(2016)]{bartling2016externality}
Bj{\"o}rn Bartling and Nick Netzer.
\newblock An externality-robust auction: Theory and experimental evidence.
\newblock \emph{Games and Economic Behavior}, 97:\penalty0 186--204, 2016.

\bibitem[Bernstein and Winter(2012)]{bernstein2012contracting}
Shai Bernstein and Eyal Winter.
\newblock Contracting with heterogeneous externalities.
\newblock \emph{American Economic Journal: Microeconomics}, 4\penalty0
  (2):\penalty0 50--76, 2012.

\bibitem[Blume et~al.(1991)Blume, Brandenburger, and
  Dekel]{blume1991lexicographic}
Lawrence Blume, Adam Brandenburger, and Eddie Dekel.
\newblock Lexicographic probabilities and choice under uncertainty.
\newblock \emph{Econometrica: Journal of the Econometric Society}, pages
  61--79, 1991.

\bibitem[Booth et~al.(2010)Booth, Chevaleyre, Lang, Mengin, and
  Sombattheera]{booth2010learning}
Richard Booth, Yann Chevaleyre, J{\'e}r{\^o}me Lang, J{\'e}r{\^o}me Mengin, and
  Chattrakul Sombattheera.
\newblock Learning conditionally lexicographic preference relations.
\newblock In \emph{ECAI}, volume~10, pages 269--274, 2010.

\bibitem[Brandt and Wei{\ss}(2001)]{brandt2001antisocial}
Felix Brandt and Gerhard Wei{\ss}.
\newblock Antisocial agents and vickrey auctions.
\newblock In \emph{International Workshop on Agent Theories, Architectures, and
  Languages}, pages 335--347. Springer, 2001.

\bibitem[Brandt et~al.(2005)Brandt, Sandholm, and Shoham]{brandt2005spiteful}
Felix Brandt, Tuomas Sandholm, and Yoav Shoham.
\newblock Spiteful bidding in sealed-bid auctions.
\newblock In \emph{Computing and Markets}, 2005.

\bibitem[Cai and Zhao(2017)]{cai2017simple}
Yang Cai and Mingfei Zhao.
\newblock Simple mechanisms for subadditive buyers via duality.
\newblock In \emph{Proceedings of the 49th Annual ACM SIGACT Symposium on
  Theory of Computing}, pages 170--183, 2017.

\bibitem[Cai et~al.(2021)Cai, Devanur, and Weinberg]{cai2019duality}
Yang Cai, Nikhil~R. Devanur, and S.~Matthew Weinberg.
\newblock A duality-based unified approach to bayesian mechanism design.
\newblock \emph{SIAM Journal on Computing}, 50\penalty0 (3):\penalty0
  STOC16--160--STOC16--200, 2021.

\bibitem[Chawla et~al.(2007)Chawla, Hartline, and
  Kleinberg]{chawla2007algorithmic}
Shuchi Chawla, Jason~D Hartline, and Robert Kleinberg.
\newblock Algorithmic pricing via virtual valuations.
\newblock In \emph{Proceedings of the 8th ACM Conference on Electronic
  Commerce}, pages 243--251, 2007.

\bibitem[Chawla et~al.(2010)Chawla, Hartline, Malec, and
  Sivan]{chawla2010multi}
Shuchi Chawla, Jason~D Hartline, David~L Malec, and Balasubramanian Sivan.
\newblock Multi-parameter mechanism design and sequential posted pricing.
\newblock In \emph{Proceedings of the forty-second ACM symposium on Theory of
  computing}, pages 311--320, 2010.

\bibitem[Christodoulou et~al.(2022)Christodoulou, Gkatzelis, and
  Schoepflin]{christodoulou2022optimal}
Giorgos Christodoulou, Vasilis Gkatzelis, and Daniel Schoepflin.
\newblock Optimal deterministic clock auctions and beyond.
\newblock In \emph{13th Innovations in Theoretical Computer Science Conference
  (ITCS 2022)}. Schloss Dagstuhl-Leibniz-Zentrum f{\"u}r Informatik, 2022.

\bibitem[Daskalakis et~al.(2014)Daskalakis, Deckelbaum, and
  Tzamos]{daskalakis2014complexity}
Constantinos Daskalakis, Alan Deckelbaum, and Christos Tzamos.
\newblock The complexity of optimal mechanism design.
\newblock In \emph{Proceedings of the twenty-fifth annual ACM-SIAM symposium on
  Discrete algorithms}, pages 1302--1318. SIAM, 2014.

\bibitem[Daskalakis et~al.(2017)Daskalakis, Deckelbaum, and
  Tzamos]{daskalakis2017strong}
Constantinos Daskalakis, Alan Deckelbaum, and Christos Tzamos.
\newblock Strong duality for a multiple-good monopolist.
\newblock \emph{Econometrica}, 85\penalty0 (3):\penalty0 735--767, 2017.

\bibitem[Diana et~al.(2021)Diana, Gill, Globus-Harris, Kearns, Roth, and
  Sharifi-Malvajerdi]{diana2021lexicographically}
Emily Diana, Wesley Gill, Ira Globus-Harris, Michael Kearns, Aaron Roth, and
  Saeed Sharifi-Malvajerdi.
\newblock Lexicographically fair learning: Algorithms and generalization.
\newblock In \emph{2nd Symposium on Foundations of Responsible Computing},
  page~1, 2021.

\bibitem[Dutting et~al.(2020)Dutting, Feldman, Kesselheim, and
  Lucier]{dutting2020prophet}
Paul Dutting, Michal Feldman, Thomas Kesselheim, and Brendan Lucier.
\newblock Prophet inequalities made easy: Stochastic optimization by pricing
  nonstochastic inputs.
\newblock \emph{SIAM Journal on Computing}, 49\penalty0 (3):\penalty0 540--582,
  2020.

\bibitem[D{\"u}tting et~al.(2020)D{\"u}tting, Kesselheim, and
  Lucier]{dutting2020log}
Paul D{\"u}tting, Thomas Kesselheim, and Brendan Lucier.
\newblock An {$O(\log \log m)$} prophet inequality for subadditive
  combinatorial auctions.
\newblock In \emph{2020 IEEE 61st Annual Symposium on Foundations of Computer
  Science (FOCS)}, pages 306--317. IEEE, 2020.

\bibitem[Fehr and Schmidt(2006)]{fehr2006economics}
Ernst Fehr and Klaus~M Schmidt.
\newblock The economics of fairness, reciprocity and altruism--experimental
  evidence and new theories.
\newblock \emph{Handbook of the economics of giving, altruism and reciprocity},
  1:\penalty0 615--691, 2006.

\bibitem[Feldman et~al.(2014)Feldman, Gravin, and
  Lucier]{feldman2014combinatorial}
Michal Feldman, Nick Gravin, and Brendan Lucier.
\newblock Combinatorial auctions via posted prices.
\newblock In \emph{Proceedings of the twenty-sixth annual ACM-SIAM symposium on
  Discrete algorithms}, pages 123--135. SIAM, 2014.

\bibitem[Feldman et~al.(2022)Feldman, Gkatzelis, Gravin, and
  Schoepflin]{feldman2022bayesian}
Michal Feldman, Vasilis Gkatzelis, Nick Gravin, and Daniel Schoepflin.
\newblock Bayesian and randomized clock auctions.
\newblock \emph{arXiv preprint arXiv:2202.09291}, 2022.

\bibitem[Fishburn(1975)]{fishburn1975axioms}
Peter~C Fishburn.
\newblock Axioms for lexicographic preferences.
\newblock \emph{The Review of Economic Studies}, 42\penalty0 (3):\penalty0
  415--419, 1975.

\bibitem[Gibbard(1973)]{gibbard1973manipulation}
Allan Gibbard.
\newblock Manipulation of voting schemes: a general result.
\newblock \emph{Econometrica: journal of the Econometric Society}, pages
  587--601, 1973.

\bibitem[Goldberg and Hartline(2005)]{goldberg2005collusion}
Andrew~V Goldberg and Jason~D Hartline.
\newblock Collusion-resistant mechanisms for single-parameter agents.
\newblock In \emph{SODA}, volume~5, pages 620--629, 2005.

\bibitem[Hagerty and Rogerson(1987)]{hagerty1987robust}
Kathleen~M Hagerty and William~P Rogerson.
\newblock Robust trading mechanisms.
\newblock \emph{Journal of Economic Theory}, 42\penalty0 (1):\penalty0 94--107,
  1987.

\bibitem[Hart and Nisan(2017)]{hart2017approximate}
Sergiu Hart and Noam Nisan.
\newblock Approximate revenue maximization with multiple items.
\newblock \emph{Journal of Economic Theory}, 172:\penalty0 313--347, 2017.

\bibitem[Hart and Reny(2015)]{hart2015maximal}
Sergiu Hart and Philip~J Reny.
\newblock Maximal revenue with multiple goods: Nonmonotonicity and other
  observations.
\newblock \emph{Theoretical Economics}, 10\penalty0 (3):\penalty0 893--922,
  2015.

\bibitem[Hatfield(2009)]{hatfield2009strategy}
John~William Hatfield.
\newblock Strategy-proof, efficient, and nonbossy quota allocations.
\newblock \emph{Social Choice and Welfare}, 33\penalty0 (3):\penalty0 505--515,
  2009.

\bibitem[Jehiel et~al.(1996)Jehiel, Moldovanu, and Stacchetti]{jehiel1996not}
Philippe Jehiel, Benny Moldovanu, and Ennio Stacchetti.
\newblock How (not) to sell nuclear weapons.
\newblock \emph{The American Economic Review}, pages 814--829, 1996.

\bibitem[Jehiel et~al.(1999)Jehiel, Moldovanu, and
  Stacchetti]{jehiel1999multidimensional}
Philippe Jehiel, Benny Moldovanu, and Ennio Stacchetti.
\newblock Multidimensional mechanism design for auctions with externalities.
\newblock \emph{Journal of economic theory}, 85\penalty0 (2):\penalty0
  258--293, 1999.

\bibitem[Kleinberg and Weinberg(2012)]{kleinberg2012matroid}
Robert Kleinberg and Seth~Matthew Weinberg.
\newblock Matroid prophet inequalities.
\newblock In \emph{Proceedings of the forty-fourth annual ACM symposium on
  Theory of computing}, pages 123--136, 2012.

\bibitem[Kohli and Jedidi(2007)]{kohli2007representation}
Rajeev Kohli and Kamel Jedidi.
\newblock Representation and inference of lexicographic preference models and
  their variants.
\newblock \emph{Marketing Science}, 26\penalty0 (3):\penalty0 380--399, 2007.

\bibitem[Kojima(2010)]{kojima2010impossibility}
Fuhito Kojima.
\newblock Impossibility of stable and nonbossy matching mechanisms.
\newblock \emph{Economics Letters}, 107\penalty0 (1):\penalty0 69--70, 2010.

\bibitem[Levine(1998)]{levine1998modeling}
David~K Levine.
\newblock Modeling altruism and spitefulness in experiments.
\newblock \emph{Review of economic dynamics}, 1\penalty0 (3):\penalty0
  593--622, 1998.

\bibitem[Li(2017)]{li2017obviously}
Shengwu Li.
\newblock Obviously strategy-proof mechanisms.
\newblock \emph{American Economic Review}, 107\penalty0 (11):\penalty0
  3257--87, 2017.

\bibitem[Li and Yao(2013)]{li2013revenue}
Xinye Li and Andrew Chi-Chih Yao.
\newblock On revenue maximization for selling multiple independently
  distributed items.
\newblock \emph{Proceedings of the National Academy of Sciences}, 110\penalty0
  (28):\penalty0 11232--11237, 2013.

\bibitem[Luss(1999)]{luss1999equitable}
Hanan Luss.
\newblock On equitable resource allocation problems: A lexicographic minimax
  approach.
\newblock \emph{Operations Research}, 47\penalty0 (3):\penalty0 361--378, 1999.

\bibitem[Makowski and Mezzetti(1994)]{makowski1994bayesian}
Louis Makowski and Claudio Mezzetti.
\newblock Bayesian and weakly robust first best mechanisms: characterizations.
\newblock \emph{Journal of Economic Theory}, 64\penalty0 (2):\penalty0
  500--519, 1994.

\bibitem[Manelli and Vincent(2007)]{manelli2007multidimensional}
Alejandro~M Manelli and Daniel~R Vincent.
\newblock Multidimensional mechanism design: Revenue maximization and the
  multiple-good monopoly.
\newblock \emph{Journal of Economic theory}, 137\penalty0 (1):\penalty0
  153--185, 2007.

\bibitem[Mirrokni et~al.(2008)Mirrokni, Schapira, and
  Vondr{\'a}k]{mirrokni2008tight}
Vahab Mirrokni, Michael Schapira, and Jan Vondr{\'a}k.
\newblock Tight information-theoretic lower bounds for welfare maximization in
  combinatorial auctions.
\newblock In \emph{Proceedings of the 9th ACM conference on Electronic
  commerce}, pages 70--77, 2008.

\bibitem[Mishra and Quadir(2014)]{mishra2014non}
Debasis Mishra and Abdul Quadir.
\newblock Non-bossy single object auctions.
\newblock \emph{Economic Theory Bulletin}, 2\penalty0 (1):\penalty0 93--110,
  2014.

\bibitem[Miyagawa(2001)]{miyagawa2001house}
Eiichi Miyagawa.
\newblock House allocation with transfers.
\newblock \emph{Journal of Economic Theory}, 100\penalty0 (2):\penalty0
  329--355, 2001.

\bibitem[Morgan et~al.(2003)Morgan, Steiglitz, and Reis]{morgan2003spite}
John Morgan, Ken Steiglitz, and George Reis.
\newblock The spite motive and equilibrium behavior in auctions.
\newblock \emph{Contributions in Economic Analysis \& Policy}, 2\penalty0
  (1):\penalty0 1--25, 2003.

\bibitem[Mukherjee(2015)]{mukherjee2015axioms}
Conan Mukherjee.
\newblock On axioms underlying use of reserve price.
\newblock \emph{Working Paper/Department of Economics, School of Economics and
  Management, Lund University}, \penalty0 (7), 2015.

\bibitem[Myerson(1981)]{myerson1981optimal}
Roger~B Myerson.
\newblock Optimal auction design.
\newblock \emph{Mathematics of operations research}, 6\penalty0 (1):\penalty0
  58--73, 1981.

\bibitem[Nath and Sen(2015)]{nath2015affine}
Swaprava Nath and Arunava Sen.
\newblock Affine maximizers in domains with selfish valuations.
\newblock \emph{ACM Transactions on Economics and Computation (TEAC)},
  3\penalty0 (4):\penalty0 1--19, 2015.

\bibitem[P{\'a}pai(2000)]{papai2000strategyproof}
Szilvia P{\'a}pai.
\newblock Strategyproof assignment by hierarchical exchange.
\newblock \emph{Econometrica}, 68\penalty0 (6):\penalty0 1403--1433, 2000.

\bibitem[P{\'a}pai(2001)]{papai2001strategyproof}
Szilvia P{\'a}pai.
\newblock Strategyproof and nonbossy multiple assignments.
\newblock \emph{Journal of Public Economic Theory}, 3\penalty0 (3):\penalty0
  257--271, 2001.

\bibitem[Pycia and Raghavan(2021)]{pycia2021non}
Marek Pycia and Madhav Raghavan.
\newblock Non-bossiness and first-price auctions.
\newblock \emph{Available at SSRN 3941784}, 2021.

\bibitem[Pycia and Troyan(2018)]{pycia2018obvious}
Marek Pycia and Peter Troyan.
\newblock Obvious dominance and random priority.
\newblock \emph{SSRN Electronic Journal}, \penalty0 (2853563), 2018.

\bibitem[Pycia and {\"U}nver(2017)]{pycia2017incentive}
Marek Pycia and M~Utku {\"U}nver.
\newblock Incentive compatible allocation and exchange of discrete resources.
\newblock \emph{Theoretical Economics}, 12\penalty0 (1):\penalty0 287--329,
  2017.

\bibitem[Pycia and {\"U}nver(2025)]{pycia2025ordinal}
Marek Pycia and M~Utku {\"U}nver.
\newblock Ordinal simplicity in discrete mechanism design.
\newblock \emph{International Economic Review}, 66\penalty0 (4):\penalty0
  1665--1680, 2025.

\bibitem[Rubinstein and Weinberg(2018)]{rubinstein2018simple}
Aviad Rubinstein and S~Matthew Weinberg.
\newblock Simple mechanisms for a subadditive buyer and applications to revenue
  monotonicity.
\newblock \emph{ACM Transactions on Economics and Computation (TEAC)},
  6\penalty0 (3-4):\penalty0 1--25, 2018.

\bibitem[Saban and Sethuraman(2014)]{saban2014note}
Daniela Saban and Jay Sethuraman.
\newblock A note on object allocation under lexicographic preferences.
\newblock \emph{Journal of Mathematical Economics}, 50:\penalty0 283--289,
  2014.

\bibitem[Satterthwaite and Sonnenschein(1981)]{satterthwaite1981strategy}
Mark~A Satterthwaite and Hugo Sonnenschein.
\newblock Strategy-proof allocation mechanisms at differentiable points.
\newblock \emph{The Review of Economic Studies}, 48\penalty0 (4):\penalty0
  587--597, 1981.

\bibitem[Satterthwaite and Williams(2002)]{satterthwaite2002optimality}
Mark~A Satterthwaite and Steven~R Williams.
\newblock The optimality of a simple market mechanism.
\newblock \emph{Econometrica}, 70\penalty0 (5):\penalty0 1841--1863, 2002.

\bibitem[Satterthwaite(1975)]{satterthwaite1975strategy}
Mark~Allen Satterthwaite.
\newblock Strategy-proofness and arrow's conditions: Existence and
  correspondence theorems for voting procedures and social welfare functions.
\newblock \emph{Journal of economic theory}, 10\penalty0 (2):\penalty0
  187--217, 1975.

\bibitem[Schmeidler and Sonnenschein(1978)]{schmeidler1978two}
David Schmeidler and Hugo Sonnenschein.
\newblock Two proofs of the gibbard-satterthwaite theorem on the possibility of
  a strategy-proof social choice function.
\newblock In \emph{Decision Theory and Social Ethics: Issues in Social Choice},
  pages 227--234. Springer, 1978.

\bibitem[Schummer(2000)]{schummer2000eliciting}
James Schummer.
\newblock Eliciting preferences to assign positions and compensation.
\newblock \emph{Games and Economic Behavior}, 30\penalty0 (2):\penalty0
  293--318, 2000.

\bibitem[Segal(1999)]{segal1999contracting}
Ilya Segal.
\newblock Contracting with externalities.
\newblock \emph{The Quarterly Journal of Economics}, 114\penalty0 (2):\penalty0
  337--388, 1999.

\bibitem[Shapley and Scarf(1974)]{shapley1974cores}
Lloyd Shapley and Herbert Scarf.
\newblock On cores and indivisibility.
\newblock \emph{Journal of mathematical economics}, 1\penalty0 (1):\penalty0
  23--37, 1974.

\bibitem[Svensson(1999)]{svensson1999strategy}
Lars-Gunnar Svensson.
\newblock Strategy-proof allocation of indivisible goods.
\newblock \emph{Social Choice and Welfare}, 16\penalty0 (4):\penalty0 557--567,
  1999.

\bibitem[Svensson(2002)]{svensson2002fixed}
Lars-Gunnar Svensson.
\newblock Strategy-proofness and fixed-price allocation of indivisible goods -
  a characterization proof.
\newblock Working Paper 2002:17, Lund, 2002.
\newblock URL \url{https://hdl.handle.net/10419/259866}.

\bibitem[Svensson and Larsson(2002)]{svensson2002strategy}
Lars-Gunnar Svensson and Bo~Larsson.
\newblock Strategy-proof and nonbossy allocation of indivisible goods and
  money.
\newblock \emph{Economic Theory}, 20\penalty0 (3):\penalty0 483--502, 2002.

\bibitem[Thanassoulis(2004)]{thanassoulis2004haggling}
John Thanassoulis.
\newblock Haggling over substitutes.
\newblock \emph{Journal of Economic theory}, 117\penalty0 (2):\penalty0
  217--245, 2004.

\bibitem[Thomson(2016)]{thomson2016non}
William Thomson.
\newblock Non-bossiness.
\newblock \emph{Social Choice and Welfare}, 47\penalty0 (3):\penalty0 665--696,
  2016.

\bibitem[Yao(2015)]{yao2015n}
Andrew Chi-Chih Yao.
\newblock An n-to-1 bidder reduction for multi-item auctions and its
  applications.
\newblock In \emph{Proceedings of the twenty-sixth annual ACM-SIAM symposium on
  Discrete algorithms}, pages 92--109, 2015.

\bibitem[Zhou and Lukose(2007)]{zhou2007vindictive}
Yunhong Zhou and Rajan Lukose.
\newblock Vindictive bidding in keyword auctions.
\newblock In \emph{Proceedings of the ninth international conference on
  Electronic commerce}, pages 141--146, 2007.

\end{thebibliography}
\end{document}